%% file: main.tex
\documentclass{llncs}
\usepackage{etex}

\usepackage{amsmath, amssymb}

\usepackage[latin1]{inputenc}
\usepackage[english]{babel}
\usepackage{times}

\usepackage{graphicx}
\usepackage{caption}
\usepackage{subcaption}
\usepackage{longtable}
\usepackage{booktabs}
\usepackage{tabu}
\usepackage[referable]{threeparttablex}

\usepackage{hyperref}
\usepackage[capitalise]{cleveref}

\crefname{figure}{Figure}{Figure}
\crefformat{footnote}{#2\footnotemark[#1]#3}
\usepackage{tikz}
\usepackage{comment}
\usepackage{parskip}
\usepackage{xfrac}

\usepackage{todonotes}
\usepackage[ruled, linesnumbered]{algorithm2e}
\usepackage{tikz}
\usetikzlibrary{arrows,automata,shapes,decorations,decorations.markings,calc, matrix,decorations.pathmorphing}

\usepackage{bm}
\usepackage{pgfplots}
\usepackage{enumerate,paralist}
\usepackage{pbox}

\let\doendproof\endproof
\renewcommand\endproof{~\hfill\qed\doendproof}

\spnewtheorem{claim}{Claim}{\bfseries}{\rmfamily}
\spnewtheorem{fact}{Fact}{\bfseries}{\rmfamily}

\DeclareMathAlphabet{\mathpzc}{OT1}{pzc}{m}{it}

\usepackage{fullpage}

\usepackage{graphicx} 
\newcommand{\Nats}{\mathbb{N}}
\newcommand{\Natsinf}{\mathbb{N}\cup\{\infty\}}
\newcommand{\Ints}{\mathbb{Z}}
\newcommand{\Intsminus}{\mathbb{Z}_{\leq 0}\cup\{-\infty\}}
\newcommand{\Intsplusminus}{\mathbb{Z}\cup\{-\infty,\infty\}}

\newcommand{\Level}{\mathsf{Lv}}

\newcommand{\LD}{\mathsf{LD}}

\newcommand{\ov}{\overline}

\newcommand{\Distance}{d}

\newcommand{\Weight}{\mathsf{wt}}

\newcommand{\Value}{\overline{\Weight}}

\newcommand{\argmin}{\mathsf{argmin}}
\newcommand{\Bag}{B}
\newcommand{\Tree}{\mathrm{Tree}}
\newcommand{\restr}{\upharpoonright}
\newcommand{\Ushape}{\mathsf{U}}

\newcommand{\Preprocessalgo}{\mathsf{Preprocess}}
\newcommand{\Updatealgo}{\mathsf{Update}}

\newcommand{\Comp}{\mathcal{C}}

\newcommand{\True}{\mathsf{True}}

\newcommand{\Cyclealgo}{\mathsf{MinCycle}}
\newcommand{\Energynodesalgo}{\mathsf{ZeroEnergyNodes}}
\newcommand{\Energynodesalgotw}{\mathsf{ZeroEnergyNodesTW}}
\newcommand{\Decisionenergyalgo}{\mathsf{DecisionEnergy}}
\newcommand{\Killcyclealgo}{\mathsf{KillCycle}}
\newcommand{\Time}{\mathcal{T}}
\newcommand{\Space}{\mathcal{S}}
\newcommand{\Energy}{\mathsf{E}}
\newcommand{\Timef}{\Weight^{\prime}}

\newcommand{\Intstriplets}{\mathcal{I}}
\newcommand{\Min}{\bm{\min}}
\newcommand{\Plus}{\bm{+}}
\newcommand{\Path}{\rightsquigarrow}
\newcommand{\Nodeset}{\mathpzc{V}}
\newcommand{\Edgeset}{\mathpzc{E}}
\newcommand{\Weightmap}{{\mathpzc{wt}}}
\newcommand{\Graphstruct}{\mathpzc{G}}

\author{Krishnendu Chatterjee$^\dag$
\qquad Rasmus Ibsen-Jensen$^\dag$ 
\qquad Andreas Pavlogiannis$^\dag$ \\[5pt]
}

\institute{$\strut^\dag$ IST Austria}

\title{Faster Algorithms for Quantitative Verification in Constant Treewidth Graphs
\thanks{The research was partly supported by Austrian Science Fund (FWF) Grant No P23499- N23, 
FWF NFN Grant No S11407-N23 (RiSE/SHiNE), ERC Start grant (279307: Graph Games), and 
Microsoft faculty fellows award.
}}

\begin{document}

\maketitle

\begin{abstract}
We consider the core algorithmic problems related to verification of systems with respect
to three classical quantitative properties, namely, the mean-payoff property, the ratio 
property, and the minimum initial credit for energy property. 
The algorithmic problem given a graph and a quantitative property asks to compute the 
optimal value (the infimum value over all traces) from every node of the graph. 
We consider graphs with constant treewidth, and it is well-known that the control-flow graphs 
of most programs have constant treewidth.
Let $n$ denote the number of nodes of a graph, $m$ the number of edges (for constant treewidth
graphs $m=O(n)$) and $W$ the largest absolute value of the weights.
Our main theoretical results are as follows.
First, for constant treewidth graphs we present an algorithm that approximates the mean-payoff value
within a multiplicative factor of $\epsilon$ in time $O(n \cdot \log (n/\epsilon))$ 
and linear space, as compared to the classical algorithms that require quadratic time.
Second, for the ratio property we present an algorithm that for constant treewidth graphs
works in time $O(n \cdot \log (|a\cdot b|))=O(n\cdot\log (n\cdot W))$, when the output is $\frac{a}{b}$, as compared to the previously best known algorithm
with running time $O(n^2 \cdot \log (n\cdot W))$.
Third, for the minimum initial credit problem we show that (i)~for general graphs the problem can be solved in $O(n^2\cdot m)$ time and the associated decision problem
can be solved in $O(n\cdot m)$ time, improving the previous known $O(n^3\cdot m\cdot \log (n\cdot W))$ and $O(n^2 \cdot m)$ bounds, respectively; and 
(ii)~for constant treewidth graphs we present an algorithm that requires $O(n\cdot \log n)$ time, improving the previous known 
$O(n^4 \cdot \log (n \cdot W))$ bound.
We have implemented some of our algorithms and show that they present a significant speedup 
on standard benchmarks. 
\end{abstract}

\input{intro}

\input{definitions}

\input{cycle}

\input{mean_cycle}

\input{initial_credit}

\input{results}

\clearpage

{
\bibliographystyle{splncs}
\bibliography{bibliography}
}

\end{document}

%% file: intro.tex
\section{Introduction}
\noindent{\bf Boolean vs quantitative verification.}
The traditional view of verification has been \emph{qualitative (Boolean)} 
that classifies traces of a system as ``correct" vs ``incorrect". 
In the recent years, motivated by applications to analyze resource-constrained
systems (such as embedded systems), there has been a huge interest to 
study \emph{quantitative} properties of systems. 
A quantitative property assigns to each trace of a system a real-number that
quantifies how good or bad the trace is, instead of classifying it as correct
vs incorrect. 
For example, a Boolean property may require that every request is eventually
granted, whereas a quantitative property for each trace can measure the 
average waiting time between requests and corresponding grants.

\noindent{\bf Variety of results.} 
Given the importance of quantitative verification, the traditional qualitative 
view of verification has been extended in several ways, such as,
quantitative languages and quantitative automata for specification 
languages~\cite{WeightedHandbook,CDH10,CDH11,CHO14,CDHRR10,Vel12,Droste}; 
quantitative logics for specification languages~\cite{BCHK11,Patricia,Orna};
quantitative synthesis for robust reactive systems~\cite{BCHJ10,BGHJ11,CHJS15};
a framework for quantitative abstraction refinement~\cite{CHR13}; 
quantitative analysis of infinite-state systems~\cite{CV12,Chatterjee15}; and
model measuring (that extends model checking)~\cite{HO14}, to name a few. 
The core algorithmic question for many of the above studies is a graph 
algorithmic problem that requires to analyze a graph wrt a quantitative
property.


\noindent{\bf Important quantitative properties.}
The three quantitative properties that have been studied for their relevance
in analysis of reactive systems are as follows.
First, the \emph{mean-payoff} property consists of a weight function that assigns 
to every transition an integer-valued weight and assigns to each trace the long-run
average of the weights of the transitions of the trace.
Second, the \emph{ratio} property consists of two weight functions (one of which is a 
positive weight function) and assigns to each trace the ratio of the two mean-payoff
properties (the denominator is wrt the positive function). 
The \emph{minimum initial credit for energy} property consists of a weight function 
(like in the mean-payoff property) and assigns to each trace the minimum number to
be added such that the partial sum of the weights for every prefix of the trace
is non-negative.
For example, the mean-payoff property is used for average waiting time, 
worst-case execution time analysis~\cite{CDH10,CHR13,Chatterjee15};
the ratio property is used in robustness analysis of systems~\cite{BGHJ11};
and the minimum initial credit for energy for measuring resource 
consumptions~\cite{Bouyer08}.

\noindent{\bf Algorithmic problems.}
Given a graph and a quantitative property, the value of a node is the infimum 
value of all traces that start at the respective node. 
The algorithmic problem (namely, the \emph{value} problem) for analysis of quantitative properties consists of a 
graph and a quantitative property, and asks to compute either the exact value 
or an approximation of the value for every node in the graph.
The algorithmic problems are at the heart of many applications, such as 
automata emptiness, model measuring, quantitative abstraction refinement, etc.

\noindent{\bf Treewidth of graphs.}
A very well-known concept in graph theory is the notion of {\em treewidth} of 
a graph, which is a measure of how similar a graph is to a tree 
(a graph has treewidth~1 precisely if it is a tree)~\cite{Robertson84}. 
The treewidth of a graph is defined based on a {\em tree decomposition} of 
the graph~\cite{Halin76}, see~\cref{sec:definitions} for a formal definition. 
Beyond the mathematical elegance of the treewidth property for graphs,
there are many classes of graphs which arise in practice and have constant treewidth. 
The most important example is that the control flow graphs of goto-free programs 
for many programming languages are of constant treewidth~\cite{Thorup98},
and it was also shown in~\cite{Gustedt02} that typically all Java programs have 
constant treewidth. 
For many other applications see the surveys~\cite{Bodlaender93,Bodlaender05}.
The constant treewidth property of graphs has also played an important role in 
logic and verification; for example, MSO (Monadic Second Order logic) queries
can be solved in polynomial time~\cite{C90} (also in log-space~\cite{Elberfeld10})
for constant-treewidth graphs;
parity games on graphs with constant treewidth can be solved in polynomial time
\cite{Obdrzalek03}; and there exist faster algorithms for probabilistic models 
(like Markov decision processes)~\cite{CL13}.
Moreover, recently it has been shown that the constant treewidth property is 
also useful for interprocedural analysis~\cite{Chatterjee15}.


\begin{table}
{\small
\renewcommand{\arraystretch}{1.2}
\centerline{
\begin{tabular}{|c|c|c|c|c|c|}
\hline
\multicolumn{3}{|c|}{Minimum mean-cycle value}  & \multicolumn{3}{c|}{Minimum ratio-cycle value} \\
\hline
\hline
Orlin \& Ahuja~\cite{OA92} & Karp~\cite{Karp78} & 
\parbox[t]{2.7cm}{\bf \centering Our result [Thm~\ref{them:min_mean_approx}] \\ ($\epsilon$-approximate)  }& Burns~\cite{B91} & Lawler~\cite{L76} &\parbox[t]{2.6cm}{\bf \centering Our result [Cor~\ref{cor:min_mean_exact}]}
\\
\hline
$O(n^{1.5}\cdot \log (n\cdot W))$ & $O(n^2)$ & $\mathbf{O(n\cdot \log (n/\epsilon))}$ & 
$O(n^3)$ & $O(n^2\cdot \log (n\cdot W))$ & $\mathbf{O(n\cdot \log (|a\cdot b|))}$
\\
\hline
\end{tabular}
}
}
\caption{Time complexity of existing and our solutions for the minimum mean-cycle value and ratio-cycle value problem in constant treewidth weighted graphs with $n$ nodes and largest absolute weight $W$, when the output is the (irreducible) fraction $\frac{a}{b}\neq 0$. 
}\label{tab:ratio}
\end{table}

\begin{table}
\renewcommand{\arraystretch}{1.2}
{\small
\centerline{
\begin{tabular}{|c|c|c|c|}
\hline
& Bouyer et. al.~\cite{Bouyer08} & \parbox[t]{2.7cm}{\bf \centering Our result  \\ {[Thm~\ref{them:energy_decision}, Cor~\ref{cor:energy}]}} & \parbox[t]{2.9cm}{\bf \centering Our result  [Thm~\ref{them:energy_tw}]  \\ (constant treewidth)}\\
\hline
\hline
Time (decision) & $O(n^2\cdot m)$ & $\mathbf{O(n\cdot m)}$ & $\mathbf{O(n\cdot \log n)}$\\
\hline
Time & $O(n^3\cdot m \cdot \log(n\cdot W))$ & $\mathbf{O(n^2 \cdot m)}$ & $\mathbf{O(n\cdot \log n)}$\\
\hline
Space & $O(n)$ & $\mathbf{O(n)}$ & $\mathbf{O(n)}$\\
\hline
\end{tabular}
}
}
\caption{Complexity of the existing and our solution for the minimum initial credit problem
on weighted graphs of $n$ nodes, $m$ edges, and largest absolute weight $W$.}\label{tab:energy}
\end{table}

\noindent{\bf Previous results and our contributions.}
In this work we consider general graphs and graphs with constant treewidth,
and the algorithmic problems to compute the exact value or an approximation
of the value for every node wrt to quantitative properties given as 
the mean-payoff, the ratio, or the minimum initial credit for energy.
We first present the relevant previous results, and then our contributions.

\noindent{\em Previous results.}
We consider graphs with $n$ nodes, $m$ edges, and let $W$ denote the largest 
absolute value of the weights. 
The running time of the algorithms is characterized by the number of arithmetic
operations (i.e., each operation takes constant time); and the space is 
characterized by the maximum number of integers the algorithm stores.
The classical algorithm for graphs with mean-payoff properties is the minimum
mean-cycle problem of Karp~\cite{Karp78}, and the algorithm 
requires $O(n\cdot m)$ running time and $O(n^2)$ space. 
A different algorithm was proposed in~\cite{Mad02} that requires $O(n\cdot m)$ running 
time and $O(n)$ space. 
Orlin and Ahuja~\cite{OA92} gave an algorithm running in time $O(\sqrt{n}\cdot m\cdot \log (n\cdot W))$.
For some special cases there exist faster approximation algorithms~\cite{CHKLR14}.
There is a straightforward reduction of the ratio problem to the 
mean-payoff problem.
For computing the exact minimum ratio, the fastest known strongly polynomial time algorithm is Burns' algorithm~\cite{B91} running in time $O(n^2\cdot m)$. Also, there is an algorithm by Lawler~\cite{L76} that uses  $O(n\cdot m\cdot \log (n\cdot W))$ time.
Many pseudopolynomial algorithms are known for the problem, with polynomial dependency on the numbers appearing in the weight function, see~\cite{DIG98}.
For the  minimum initial credit for energy problem, the decision problem (i.e., is the energy required for node $v$ at most $c$?) can be solved in $O(n^2\cdot m)$ time, leading to an $O(n^3\cdot m\cdot \log (n\cdot W))$ time algorithm for the minimum initial credit for energy problem~\cite{Bouyer08}.
All the above algorithms are for general graphs (without the constant-treewidth restriction).

\noindent{\em Our contributions.} Our main contributions are as follows.
\begin{compactenum}
\item \emph{Finding the mean-payoff and ratio values in constant-treewidth graphs.} 
We present two results for constant treewidth graphs.
First, for the exact computation we present an algorithm that requires $O(n\cdot \log (|a\cdot b|))$ time 
and $O(n)$ space, where $\frac{a}{b}\neq 0$ is the (irreducible) ratio/mean-payoff of the output. 
If $\frac{a}{b}=0$ then the algorithm uses $O(n)$ time. Note that $\log (|a\cdot b|)\leq 2\log (n\cdot W)$. 
We also present a space-efficient version of the algorithm that requires only $O(\log n)$ space. 
Second, we present an algorithm for finding an $\epsilon$-factor approximation that requires $O(n\cdot\log (n/\epsilon))$ time and 
$O(n)$ space, as compared to the $O(n^{1.5}\cdot \log (n\cdot W))$ time solution of Orlin \& Ahuja,
and the $O(n^2)$ time solution of Karp (see \cref{tab:ratio}).

\item \emph{Finding the minimum initial credit in graphs.} 
We present two results.
First, we consider the exact computation for general graphs, and 
present (i)~an $O(n \cdot m)$ time algorithm for the decision problem 
(improving the previous known  $O(n^2\cdot m)$ bound), and 
(ii)~an $O(n^2\cdot m)$ time algorithm to compute value of all nodes
(improving the previous known $O(n^3\cdot m \cdot \log(n\cdot W))$ bound).
Finally, we consider the computation of the exact value for graphs with constant 
treewidth and present an algorithm that requires $O(n\cdot \log n)$ time 
(improving the previous known  $O(n^4 \cdot \log(n\cdot W))$ bound) 
(see \cref{tab:energy}).

\item \emph{Experimental results.}
We have implemented our algorithms for the minimum mean cycle and minimum initial credit
problems and ran them on standard benchmarks 
(DaCapo suit~\cite{Blackburn06} for the minimum mean cycle problem, and
DIMACS challenges~\cite{dimacs} for the minimum initial credit problem).
For the minimum mean cycle problem, our results show that our algorithm
has lower running time than all the classical polynomial-time algorithms.
For the minimum initial credit problem, our algorithm provides a significant 
speedup 
over the existing method. 
Both improvements are demonstrated even on graphs of small/medium size.
Note that our theoretical improvements (better asymptotic bounds)
imply improvements for large graphs, and our improvements on medium size
graphs indicate that our algorithms have practical applicability with 
small constants.

\end{compactenum}

\noindent{\bf Technical contributions.}
The key technical contributions of our work are as follows:

\begin{compactenum}
\item \emph{Mean-payoff and ratio values in constant-treewidth graphs.}
Given a graph with constant treewidth, let $c^{\ast}$ be the smallest weight of a simple cycle.
First, we present a linear-time algorithm that computes $c^{\ast}$ exactly 
(if $c^{\ast}\geq 0$) or approximate within a polynomial factor (if $c^{\ast}<0$).
Then, we show that if the minimum ratio value $\nu^{\ast}$ is the irreducible fraction 
$\frac{a}{b}$, then $\nu^{\ast}$ can be computed by evaluating $O(\log (|a\cdot b|))$ inequalities 
of the form $\nu^{\ast} \geq \nu$. Each such inequality is evaluated by computing the smallest weight 
of a simple cycle in a modified graph. 
Finally, for $\epsilon$-approximating the value $\nu^{\ast}$, we show that 
$O(\log (n/\epsilon))$ such inequalities suffice.

\item \emph{Minimum initial credit problem.}
We show that for general graphs, the decision problem can be solved 
with two applications of Bellman-Ford-type algorithms, and the value problem
reduces to finding non-positive cycles in the graph, followed by one instance of the 
single-source shortest path problem.
We then show how the invariants of the algorithm for the value problem on general graphs can be
maintained by a particular graph traversal of the tree-decomposition for constant-treewidth graphs.
\end{compactenum}

%% file: definitions.tex
\section{Definitions}\label{sec:definitions}
\noindent{\bf Weighted graphs.}
We consider \emph{finite weighted directed graphs} $G=(V,E,\Weight, \Timef)$ where $V$ is the set of $n$ \emph{nodes},
$E\subseteq V\times V$ is the edge relation of $m$ \emph{edges}, $\Weight:E \rightarrow \Ints$
is a \emph{weight function} that assigns an integer weight $\Weight(e)$ to each edge $e\in E$,
and $\Timef:E\rightarrow \Nats^{+}$ is a weight function that assigns strictly positive integer weights.
For technical simplicity, we assume that there exists at least one outgoing edge from every node.
In certain cases where the function $\Timef$ is irrelevant, we
will consider weighted graphs $G=(V,E,\Weight)$, i.e., without the function $\Timef$.

\noindent{\bf Finite and infinite paths.}
A \emph{finite path} $P=(u_1,\dots, u_j)$, is a sequence of nodes $u_i\in V$ such that
for all $1\leq i<j$ we have $(u_i,u_{i+1})\in E$. The \emph{length} of $P$ is $|P|=j-1$.
A single-node path has length $0$.
The path $P$ is \emph{simple} if there is no node repeated in $P$, and it is a \emph{cycle} if $j>1$ and $u_1=u_j$.
The path $P$ is a \emph{simple cycle} if $P$ is a cycle and the sequence $(u_2,\dots u_j)$
is a simple path. 
The functions $\Weight $ and $\Timef$ naturally extend to paths, so that the weight of
a path $P$ with $|P|>0$ wrt the weight functions $\Weight$ and $\Timef$ is
$\Weight(P)=\sum_{1\leq i<j} \Weight(u_i,u_{i+1})$ and
$\Timef(P)=\sum_{1\leq i<j} \Timef(u_i,u_{i+1})$.
The \emph{value} of $P$ is defined to be $\Value(P)=\frac{\Weight(P)}{\Timef(P)}$.
For the case where $|P|=0$, we define $\Weight(P)=0$, and $\Value(P)$ is undefined.
An \emph{infinite path}  $\mathcal{P}=(u_1, u_2, \dots)$ of $G$ is an infinite sequence of nodes such that
every finite prefix $P$ of $\mathcal{P}$ is a finite path of $G$.
The functions $\Weight$ and $\Timef$ assign to $\mathcal{P}$
a value in $\Intsplusminus$: we have
$\Weight(\mathcal{P})=\sum_i \Weight(u_i,u_{i+1})$ and
$\Timef(\mathcal{P})=\infty$.
For a (possibly infinite) path $P$, we use the notation $u\in P$ 
to denote that a node $u$ appears in $P$,
and $e\in P$ to denote that an edge $e$ appears in $P$.
Given a set $\Bag\subseteq V$, we denote with $P\cap \Bag$ the set 
of nodes of $\Bag$ that appear in $P$.
Given a finite path $P_1$ and a possibly infinite path $P_2$,
we denote with $P_1\circ P_2$ the path resulting from the concatenation
of $P_1$ and $P_2$.

\noindent{\bf Distances and witness paths.}
For nodes $u,v\in V$, we denote with $\Distance(u,v)=\inf_{P:u\Path v} \Weight(P)$
the \emph{distance} from $u$ to $v$. 
A finite path $P:u\Path v$ is a \emph{witness} of the distance $\Distance(u,v)$
if $\Weight(P)=\Distance(u,v)$. 
An infinite path $\mathcal{P}$ is a witness of the distance $\Distance(u,v)$ if
the  following conditions hold:
\begin{compactenum}
\item $\Distance(u,v)=\Weight(\mathcal{P})=-\infty$, and
\item $\mathcal{P}$ starts from $u$, and $v$ is reachable from every node of $\mathcal{P}$.
\end{compactenum}
Note that $\Distance(u,v)=\infty$ is not witnessed by any path.

\noindent{\bf Tree decompositions.}
A \emph{tree-decomposition} $\Tree(G)=(V_T, E_T)$ of $G$ is a tree such that
the following conditions hold:
\begin{compactenum}
\item $V_T=\{\Bag_0,\dots, \Bag_{n'-1}: \forall i~\Bag_i\subseteq V\}$ and $\bigcup_{\Bag_i\in V_T}\Bag_i=V$ (every node is covered).
\item For all $(u,v)\in E$ there exists $\Bag_i\in V_T$ such that $u,v\in \Bag_i$ (every edge is covered).
\item For all $i,j,k$ such that there is a bag $\Bag_k$ that appears in the simple path $\Bag_i\Path \Bag_j$ in $\Tree(G)$, 
we have $\Bag_i\cap \Bag_j\subseteq \Bag_k$ (every node appears in a contiguous subtree of $\Tree(G)$).
\end{compactenum}
The sets $\Bag_i$ which are nodes in $V_T$ are called \emph{bags}.
Conventionally, we call $\Bag_0$ the root of $\Tree(G)$, and denote with $\Level(\Bag_i)$
the level of $\Bag_i$ in $\Tree(G)$, with $\Level(\Bag_0)=0$.
We say that $\Tree(G)$ is \emph{balanced} if the maximum level is $\max_{\Bag_i}\Level(\Bag_i)=O(\log n')$,
and it is \emph{binary} if every bag has at most two children bags.
A bag $\Bag$ is called the \emph{root bag} of a node $u$ if $\Bag$ is the smallest-level
bag that contains $u$, and we often use $\Bag_u$ to refer to the root bag of $u$.
The \emph{width} of a tree-decomposition $\Tree(G)$ is the size of the largest bag minus $1$.
The treewidth of $G$ is the smallest width among the widths of all possible tree decompositions of $G$.
The following lemma gives a fundamental structural property of tree-decompositions.

\begin{lemma}\label{lem:separator_property}
Consider a graph $G=(V,E)$, a binary tree-decomposition $T=\Tree(G)$ and a bag $\Bag$ of $T$.
Denote with $(\Comp_i)_{1\leq i \leq 3}$ the components of $T$ created by removing $\Bag$ from $T$,
and let $V_i$ be the set of nodes that appear in bags of component $\Comp_i$.
For every $i\neq j$, nodes $u\in V_i$, $v\in V_j$ and $P:u\Path v$,
we have that $P\cap\Bag\neq\emptyset$  
(i.e., all paths between $u$ and $v$ go through some node in $\Bag$).
\end{lemma}

\begin{theorem}\label{them:tree_decomp}
For every graph $G$ with $n$ nodes and constant treewidth, a balanced binary tree-decomposition $\Tree(G)$ of constant width and $O(n)$ bags can be constructed in
(1)~$O(n)$ time and space~\cite{Bodlaender95},
(2)~deterministic logspace (and hence polynomial time)~\cite{Elberfeld10}.
\end{theorem}

In the sequel we consider only balanced and binary tree-decompositions of constant width and $n'=O(n)$ bags
(and hence of height $O(\log n)$).
Additionally, we consider that every bag is the root bag of at most one node.
Obtaining this last property is straightforward, simply by replacing each bag $\Bag$ which is the root of $k>1$ nodes $x_1,\dots x_k$ with a chain of bags $\Bag_1, \dots, \Bag_k=\Bag$, where each $\Bag_i$ is the parent of $\Bag_{i+1}$,
and $\Bag_{i+1}=\Bag_i\cup\{x_{i+1}\}$. Note that this keeps the tree binary and increases its height by at most a constant factor, hence the resulting tree is also balanced.

Throughout the paper, we follow the convention that the maximum and minimum of the empty set
is $-\infty$ and $\infty$ respectively, i.e., $\max(\emptyset)=-\infty$ and $\min(\emptyset)=\infty$.
Time complexity is measured in number of arithmetic and logical operations, 
and space complexity is measured in number of machine words.
Given a graph $G$, we denote with $\Time(G)$ and $\Space(G)$ the time and space required
for constructing a balanced, binary tree-decomposition $\Tree(G)$.
We are interested in the following problems.

\noindent{\bf The minimum mean cycle problem~\cite{Karp78}.}
Given a weighted directed graph $G=(V,E,\Weight)$, the minimum mean cycle problem asks to determine
for each node $u$ the \emph{mean value} $\mu^{\ast}(u)=\min_{C\in \mathcal{C}_u}\frac{\Weight(C)}{|C|}$, 
where $\mathcal{C}_u$ is the set of simple cycles reachable from $u$ in $G$. 
A cycle $C$ with $\frac{\Weight(C)}{|C|}=\mu^{\ast}(u)$ is called a minimum mean cycle of $u$.
For $0<\epsilon<1$, we say that a value $\mu$ is an $\epsilon$-approximation of
the mean value $\mu^{\ast}(u)$ if $|\mu-\mu^{\ast}(u)|\leq \epsilon\cdot |\mu^{\ast}(u)|$.

\noindent{\bf The minimum ratio cycle problem~\cite{Hartmann93}.}
Given a weighted directed graph $G=(V,E,\Weight, \Timef)$,
the minimum ratio cycle problem asks to determine for each node $u$ the \emph{ratio value}
$\nu^{\ast}(u)=\min_{C\in\mathcal{C}_u}\Value(C)$,
where $\Value(C)=\frac{\Weight(C)}{\Timef(C)}$ and  $\mathcal{C}_u$ is the set of simple cycles reachable from $u$ in $G$. 
A cycle $C$ with $\Value(C)=\nu^{\ast}_u$ is called a minimum ratio cycle of $u$.
The minimum mean cycle problem follows as a special case of the minimum ratio cycle problem
for $\Timef(e)=1$ for each edge $e\in E$.

\noindent{\bf The minimum initial credit problem~\cite{Bouyer08}.}
Given a weighted directed graph $G=(V,E,\Weight)$, the minimum initial credit value problem
asks to determine for each node $u$ the smallest energy value $\Energy(u)\in \Natsinf$
with the following property: there exists an infinite path $\mathcal{P}=(u_1, u_2\dots)$ with $u=u_1$,
such that for every finite prefix $P$ of $\mathcal{P}$ we have $\Energy(u)+\Weight(P)\geq 0$.
Conventionally, we let $\Energy(u)=\infty$ if no finite value exists.
The associated decision problem asks given a node $u$ and an initial credit $c\in \Nats$
whether $\Energy(u)\leq c$.

%% file: cycle.tex
\section{Minimum Cycle}\label{sec:min_cycle}



In the current section we deal with a related graph problem, namely the detection of a minimum-weight simple cycle of a graph.
In~\cref{sec:mean_cycle} we use solutions to the minimum cycle problem to obtain the minimum ratio and minimum mean values
of a graph.

\noindent{\bf The minimum cycle problem.}
Given a weighted graph $G=(V,E, \Weight)$, the minimum cycle problem
asks to determine the weight $c^{\ast}$ of a minimum-weight simple cycle in $G$,
i.e., $c^{\ast}=\min_{C\in \mathcal{C}} \Weight(C)$, where $\mathcal{C}$ is the set of simple cycles in $G$.

We describe the algorithm $\Cyclealgo$ that operates
on a tree-decomposition $\Tree(G)$ of an input graph $G$, and has the following properties.

\begin{compactenum}
\item If $G$ has no negative cycles, then $\Cyclealgo$ returns the weight  $c^{\ast}$ of a minimum-weight cycle in $G$.
\item If $G$ has negative cycles, then $\Cyclealgo$ returns a value that is at most a polynomial (in $n$) factor
smaller than $c^{\ast}$.
\end{compactenum}

\noindent{\bf $\Ushape$-shaped paths.}
Following the recent work of~\cite{Chatterjee15}, we define the important notion of
$\Ushape$-shaped paths in a tree-decomposition $\Tree(G)$.
Given a bag $\Bag$ and nodes $u,v\in\Bag$, we say that a path $P:u\Path v$
is \emph{$\Ushape$-shaped} in $\Bag$, if one of the following conditions hold:
\begin{compactenum}
\item Either $|P|>1$ and for all intermediate nodes $w\in P$, we have $\Bag$ is an ancestor of $\Bag_w$,
\item or $|P|\leq 1$ and $\Bag$ is $\Bag_u$ or $\Bag_v$ (i.e., $\Bag$ is the root bag of either $u$ or $v$).
\end{compactenum}
Informally, given a bag $\Bag$, a $\Ushape$-shaped path in $\Bag$
is a path that traverses intermediate nodes that exist only in the subtree of $\Tree(G)$ rooted in $\Bag$. 
The following remark follows from the definition of tree-decompositions, and
states that every simple cycle $C$ can be seen as a $\Ushape$-shaped path 
$P$ from the smallest-level node of $C$ to itself.
Consequently, we can determine the value $c^{\ast}$  by only considering
$\Ushape$-shaped paths in $\Tree(G)$.

\begin{remark}\label{rem:ushape_cycle}
Let $C=(u_1,\dots, u_k)$ be a simple cycle in $G$, and $u_j=\arg\min_{u_i\in C}\Level(u_i)$.
Then $P=(u_j, u_{j+1},\dots u_k, u_1,\dots, u_j)$ is a $\Ushape$-shaped path in $\Bag_{u_j}$,
and $\Weight(P)=\Weight(C)$.
\end{remark}

\noindent{\bf Informal description of $\Cyclealgo$.}
Based on $\Ushape$-shaped paths, the work of~\cite{Chatterjee15} presented a method for 
computing algebraic path properties on tree-decompositions with constant width,
where the weights of the edges come from a general semiring. 
Note that integer-valued weights are a special case of the tropical semiring.
Our algorithm $\Cyclealgo$ is similar to the algorithm $\Preprocessalgo$ from~\cite{Chatterjee15}.
It consists of a depth-first traversal of $\Tree(G)$, and for each examined bag $\Bag$
computes a \emph{local distance} map $\LD_{\Bag}:\Bag\times \Bag\rightarrow \Ints\cup\{\infty\}$ 
such that for each $u,v\in \Bag$, we have (i)~$\LD_{\Bag}(u,v)=\Weight(P)$ for some path $P:u\Path v$, 
and (ii)~$\LD_{\Bag}\leq \min_{P}\Weight(P)$, where $P$ are taken to be simple $u\Path v$ paths (or simple cycles)
that are $\Ushape$-shaped in $\Bag$.
This is achieved by traversing $\Tree(G)$ in post-order, and for each
root bag $\Bag_x$ of a node $x$, we update $\LD_{\Bag_x}(u,v)$
with $\LD_{\Bag_x}(u,x) + \LD_{\Bag_x}(x,v)$
(i.e., we do path-shortening from node $u$ to node $v$, by considering paths that go through $x$).
See \cref{fig:fig_ushape} for an illustration.

In the end, $\Cyclealgo$ returns $\min_{x}\LD_{\Bag_x}(x,x)$, i.e.,
the weight of the smallest-weight $\Ushape$-shaped (not necessarily simple) cycle it has discovered.
Algorithm \ref{algo:mincycle} gives $\Cyclealgo$ in pseudocode.
For brevity, in line~\ref{line:children} we consider that if $\{u,v\}\not\in E$ or 
$\{u,v\}\not \subseteq \Bag_i$  for some child $\Bag_i$ of $\Bag$, then $\LD_{\Bag_i}(u,v)=\infty$.

\begin{algorithm}
\small
\DontPrintSemicolon
\caption{$\Cyclealgo$}\label{algo:mincycle}
\KwIn{A weighted graph $G=(V,E,\Weight)$ and a balanced binary tree-decomposition $\Tree(G)$}
\KwOut{A value $c$}
\BlankLine
Assign $c\leftarrow \infty$\\
Apply a post-order traversal on $\Tree(G)$, and examine each bag $\Bag$ with children $\Bag_1, \Bag_2$\label{line:dfs}\\
\Begin{
\ForEach{$u, v\in \Bag$}{
Assign $\LD_{\Bag}(u,v)\leftarrow\min(\LD_{\Bag_1}(u,v), \LD_{\Bag_2}(u,v), \Weight(u,v))$\label{line:children}\\
}
Discard $\LD_{\Bag_1}, \LD_{\Bag_2}$\\
\uIf{$\Bag$ is the root bag of a node $x$}{
\ForEach{$u, v\in \Bag$}{
Assign $\LD'_{\Bag}(u,v)\leftarrow \min(\LD_{\Bag}(u,v), \LD_{\Bag}(u,x)+\LD_{\Bag}(x,v))$\label{line:update}\\
}
Assign $\LD_{\Bag}\leftarrow \LD'_{\Bag}$\\
Assign $c\leftarrow \min(c, \LD_{\Bag}(x,x))$\label{line:c}\\
}
}
\Return $c$
\end{algorithm}

\input{fig_ushape}

In essence, $\Cyclealgo$ performs repeated summarizations of paths in $G$.
The following lemma follows easily from~\cite[Lemma~2]{Chatterjee15},
and states that $\LD_{\Bag}(u,v)$ is upper bounded by the smallest weight of a
$\Ushape$-shaped simple $u\Path v$ path in $\Bag$.

\begin{lemma}[{\cite[Lemma~2]{Chatterjee15}}]\label{lem:ushape}
For every examined bag $\Bag$ and nodes $u,v\in \Bag$, we have
\begin{compactenum}
\item  $\LD_{\Bag}(u,v)=\Weight(P)$ for some path $P:u\Path v$ (and $\LD_{\Bag}(u,v)=\infty$ if no such $P$ exists),
\item  $\LD_{\Bag}(u,v)\leq \min_{P:u\Path v}\Weight(P)$ where $P$ ranges over $\Ushape$-shaped simple paths and simple cycles in $\Bag$.
\end{compactenum}
\end{lemma}

At the end of the computation, the returned value $c$ is the weight of a
(generally non-simple) cycle $C$, captured as a $\Ushape$-shaped path on its smallest-level node.
The cycle $C$ can be recovered by tracing 
backwards the updates of line~\ref{line:update} performed by the algorithm,
starting from the node $x$ that performed the last update in line~\ref{line:c}.
Hence, if $C$ traverses $k$ distinct edges, we can write

\begin{equation}\label{eq:c}
c=\Weight(C)=\sum_{i=1}^{k} k_i\cdot \Weight(e_i)
\end{equation}

where each $e_i$ is a distinct edge, and $k_i$ is the number of times it appears in $C$.

\begin{lemma}\label{lem:edge_multiplicity}
Let $h$ be the height of $\Tree(G)$. For every $k_i$ in~\cref{eq:c}, we have $k_i\leq 2^h$.
\end{lemma}
\begin{proof}
Note that the edge $e_i=(u_i,v_i)$ is first considered by $\Cyclealgo$ in the root bag $\Bag_i$ of
node $x_i$, where $x_i=\arg\max_{y_i\in\{u_i,v_i\}}\Level(y_i)$ (line~\ref{line:update}). 
As $\Cyclealgo$ backtracks from $\Bag_i$ to the root of $\Tree(G)$,
the edge $e_i$ can be traversed at most twice as many times in each step (because of line~\ref{line:update}, once for each term
of the sum $\LD_{\Bag}(u,x)+\LD_{\Bag}(x,v)$). 
Hence, this doubling will occur at most $h$ times, and $k_i\leq 2^h$.
\end{proof}

\begin{lemma}\label{lem:cyclealgo_correctness}
Let $c$ be the value returned by $\Cyclealgo$, $h$ be the height of $\Tree(G)$,
and $c^{\ast}=\min_C \Weight(C)$ over all simple cycles $C$ in $G$.
The following assertions hold:
\begin{compactenum}
\item \label{item:no_neg_cycle} If $G$ has no negative cycles, then $c=c^{\ast}$.
\item If $G$ has a negative cycle, then
\begin{compactenum}
\item \label{item:neg_cycle1} $c\leq c^{\ast}$.
\item \label{item:neg_cycle2}  $|c|=O\left(|c^{\ast}|\cdot n\cdot 2^h\right)$.
\end{compactenum}
\end{compactenum}
\end{lemma}
\begin{proof}
By \cref{rem:ushape_cycle}, we have that $c^{\ast}=\Weight(P)$ for a $\Ushape$-shaped path $P:x\Path x$.
By \cref{lem:ushape}, after $\Cyclealgo$ examines $\Bag_x$, it will be $c\leq \LD_{\Bag_x}(x,x)\leq c^{\ast}$,
with the equalities holding if there are no negative cycles in $G$ (by the definition of $c^{\ast}$, as then $\LD_{\Bag_x}(x,x)$ is witnessed by a simple cycle).
By line~\ref{line:update}, $c$ can only decrease afterwards, and again by the definition of $c^{\ast}$ this can
only happen if there are negative cycles in $G$.
This proves items~\ref{item:no_neg_cycle} and \ref{item:neg_cycle1},
and the remaining of the proof focuses on showing that $|c|=O\left(|c^{\ast}|\cdot n\cdot 2^h\right)$.

By rearranging the sum of \cref{eq:c}, we can decompose the obtained cycle
$C$ into a set of $k'^{+}$ non-negative cycles $C^{+}_i$, and a set of $k'^{-}$ negative cycles $C^{-}_i$,
and each cycle $C^{+}_i$ and $C^{-}_i$ appears with multiplicity $k^{+}_i$ and $k^{-}_i$ respectively.
Then we have

\begin{align}
|c|&=|\Weight(C)| = \left|\sum_{i=1}^{k'^{+}}k^{+}_i\cdot \Weight(C^{+}_i) + \sum_{i=1}^{k'^{-}}k^{-}_i\cdot \Weight(C^{-}_i)\right| \leq  \left|\sum_{i=1}^{k'^{-}}k^{-}_i\cdot \Weight(C^{-}_i)\right| \notag\\
&\leq \sum_{i=1}^{k^{-}}k^{-}_i\cdot |\Weight(C^{-}_i)|\leq |c^{\ast}|\cdot \sum_{i=1}^{k'^{-}}k^{-}_i \leq |c^{\ast}|\cdot \sum_{i=1}^{k}k^{-}_i =O\left(|c^{\ast}|\cdot n\cdot 2^h\right)
\end{align}
The first inequality follows from $c<0$, the third inequality holds by the definition of $c^{\ast}$,
and the last inequality holds since $k'^{-}\leq k$.
Finally, we have $\sum_{i=1}^{k} k^{-}_i=O\left(n\cdot 2^h\right)$,
since $k=O(n)$, and by \cref{lem:edge_multiplicity} we have $k^{-}_i\leq 2^h$.
\end{proof}

Next we discuss the time and space complexity of $\Cyclealgo$.

\begin{lemma}\label{lem:cyclealgo_complexity}
Let $h$ be the height of $\Tree(G)$.
$\Cyclealgo$ accesses  each bag of $\Tree(G)$ a constant number of times,
and uses $O(h)$ additional space.
\end{lemma}
\begin{proof}
$\Cyclealgo$ accesses each bag a constant number of times, as it performs a 
post-order traversal  on $\Tree(G)$ (line~\ref{line:dfs}).
Because it computes the local distances in a postorder manner,
the number of local distance maps $\LD_{\Bag}$ it remembers is bounded
by the height $h$ of $\Tree(G)$. 
Since $\Tree(G)$ has constant width, $\LD_{\Bag}$ requires a constant number of words
for storing a constant number of nodes and weights in each $\Bag$.
Hence the total space usage is $O(h)$, and the result follows.
\end{proof}

The following theorem summarizes the results of this section.

\begin{theorem}\label{them:min_cycle}
Let $G=(V,E,\Weight)$ be a weighted graph of $n$ nodes with constant treewidth,
and a balanced, binary tree-decomposition $\Tree(G)$ of $G$ be given.
Let $c^{\ast}$, be the smallest weight of a simple cycle in $G$.
Algorithm $\Cyclealgo$ uses $O(n)$ time and $O(\log n)$ additional space, and returns a value $c$ such that:
\begin{compactenum}
\item \label{item:them_no_neg_cycle} If $G$ has no negative cycles, then $c=c^{\ast}$.
\item If $G$ has a negative cycle, then
\begin{compactenum}
\item \label{item:them_neg_cycle1} $c\leq c^{\ast}$.
\item \label{item:them_neg_cycle2} $|c|=|c^{\ast}|\cdot n^{O(1)}$.
\end{compactenum}
\end{compactenum}
\end{theorem}

%% file: fig_ushape.tex
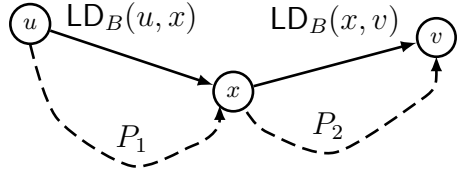
\begin{figure}
\centering
\begin{tikzpicture}[scale=0.9,
thick, >=latex,
pre/.style={<-,shorten >= 1pt, shorten <=1pt, thick},
post/.style={->,shorten >= 1pt, shorten <=1pt,  thick},
und/.style={very thick, draw=gray},
bag/.style={ellipse, minimum height=9mm,minimum width=18mm,draw=gray!80, line width=1pt},
node/.style={circle, minimum size=2mm, draw=black!100, line width=1pt, inner sep=3},
rootbag/.style={ellipse, minimum height=7mm,minimum width=14mm,draw=black!80, line width=2.5pt},
virt/.style={circle,draw=black!50,fill=black!20, opacity=0},
scale=1,every node/.style={scale=1}]

\node	[node]		(u)	at	(-3, 0)	{$u$};
\node	[virt]		(u2)	at	(-3, -0.3)	{};
\node	[node]		(x)	at	(0, -1)	{$x$};
\node	[virt]		(x3)	at	(-0.2, -1.6)	{};
\node	[virt]		(x4)	at	(0.2, -1.3)	{};
\node	[node]		(v)	at	(3, -0.2)	{$v$};
\node	[virt]		(v2)	at	(3, -0.88)	{};

\node	[virt]		(y3)	at	(-0.05, -1.3)	{};
\node	[virt]		(y4)	at	(0.05, -1.6)	{};

\node	[]		(P1)	at	(-1.5, -1.7)	{\large $P_1$};
\node	[]		(P2)	at	(1.4, -1.5)	{\large $P_2$};

\node	[]		(a1)	at	(-2.6, -1.4)	{};
\node	[]		(a2)	at	(-0.6,-1.8)	{};
\node	[]		(a3)	at	(-1.5,-2.1)	{};

\node	[]		(b2)	at	(2.8, -1.1)	{};
\node	[]		(b1)	at	(0.5,-1.6)	{};
\node	[]		(b3)	at	(1.5,-1.9)	{};


\draw [draw=black, line width =1, ->, dashed, dash pattern=on 2mm off 1mm] plot [smooth, ] coordinates {(u2) (a1) (a3) (a2) (x3) };
\draw [draw=black, line width =1,dashed, dash pattern=on 2mm off 1mm, ->] plot [smooth, ] coordinates {(x4) (b1) (b3) (b2) (v2) };

\draw [draw=black, line width =1, ->] (u) to node[above=0.2] {\large $\LD_{\Bag}(u,x)$} (x);
\draw [draw=black, line width =1, ->] (x) to node[above=0.25]{\large $\LD_{\Bag}(x,v)$} (v);

\end{tikzpicture}
\caption{Path shortening in line~\ref{line:update} of $\Cyclealgo$.
When $\Bag_x$ is examined, $\LD_{\Bag_x}(u,v)$ is updated with the weight of the $\Ushape$-shaped path $P=P_1\circ P_2$.
The paths $P_1$ and $P_2$ are $\Ushape$-shaped paths in the children bags $\Bag_1$ and $\Bag_2$, and we have $\LD_{\Bag_i}(u,x)=\Weight(P_i)$.}\label{fig:fig_ushape}
\end{figure}

%% file: mean_cycle.tex
\section{The Minimum Ratio and Mean Cycle Problems}\label{sec:mean_cycle}
In the current section we present algorithms for solving the minimum ratio and mean cycle problems
for weighted graphs $G=(V,E,\Weight, \Timef)$ of constant treewidth.

\begin{remark}\label{rem:scc}
If $G$ is not strongly connected we can compute its strongly connected components
in linear time~\cite{Tarjan72}, and use the algorithms of this section
to compute the minimum cycle ratio $\nu^{\ast}_i$ in every component $\mathcal{G}_i$ separately.
Afterwards, we compute $\nu^{\ast}(u)$ for every node $u$
by iteratively (i)~finding the nodes $u$ that can reach the component $\mathcal{G}_j$ where $j=\argmin_i \nu^{\ast}_i$,
(ii)~assigning $\nu^{\ast}(u) = \nu^{\ast}_j$, and (iii)~removing $\mathcal{G}_j$ and repeating.
Since these operations require linear time, they do not impact the time complexity.
\end{remark}

In light of \cref{rem:scc}, we consider graphs that are strongly
connected, and hence it follows that $\nu^{\ast}(u)$ is the same for every node 
$u$, and thus we will speak about the minimum ratio $\nu^{\ast}$ and mean $\mu^{\ast}$
values of $G$.

\begin{claim}\label{claim:value}
Let $\nu^{\ast}$ be the ratio value of $G$. 
Then $\nu^{\ast}\geq \nu$ iff for every cycle $C$ of $G$
we have $\Weight_{\nu}(C)\geq 0$, 
where $\Weight_{\nu}(e)=\Weight(e)-\Timef(e)\cdot \nu$ for each edge $e\in E$.
\end{claim}
\begin{proof}
Indeed, for any cycle $C$ we have $\Value(C)\geq \nu^{\ast} \geq \nu$. Then
\begin{align*}
&\Value(C)\geq \nu \iff \Value(C)-\nu \geq 0 \iff \frac{\Weight(C) - \nu\cdot \Timef(C) }{\Timef(C)}\geq 0\\
\iff &  \Weight(C) - \nu\cdot \Timef(C)\geq 0 \iff \sum_{e\in C}(\Weight(e)-\Timef(e)\cdot \nu) \geq 0 \iff  \Weight_{\nu}(C) \geq 0
\end{align*}

with the equality holding iff $\Value(C)=\nu$.
\end{proof}

Hence, given a tree-decomposition $\Tree(G)$,
for any guess $\nu$ of the ratio value $\nu^{\ast}$, we can evaluate whether $\nu^{\ast}\geq \nu$ by
constructing the weight function $\Weight_{\nu}=\Weight-\nu$ and executing algorithm $\Cyclealgo$ on input $G_{\nu}=(V,E,\Weight_{\nu})$.
By \cref{item:them_neg_cycle1} of \cref{them:min_cycle} and Claim~\ref{claim:value} we have that
the returned value $c$ of $\Cyclealgo$ is $c\geq 0$ iff $\Weight_{\nu}(C)\geq 0$ for all cycles $C$, iff $\nu^{\ast}\geq \nu$
(and in fact $c=0$ iff $\nu^{\ast} = \nu$).

\begin{lemma}\label{lem:ratio_decision}
Let $G=(V,E,\Weight, \Timef)$ be a weighted graph of $n$ nodes with constant treewidth and minimum ratio value $\nu^{\ast}$.
Let $\Tree(G)$ be a given balanced, binary tree-decomposition of $G$ of constant width.
For any rational $\nu$, the decision problem of whether $\nu^{\ast} \geq \nu$ (or $\nu^{\ast} = \nu$) can be solved 
in $O(n)$ time and $O(\log n)$ extra space.
\end{lemma}
\begin{proof}
By Claim~\ref{claim:value}, we can test whether $\nu^{\ast} \geq \nu$
by testing whether $G_{\nu}=(V,E,\Weight_{\nu})$ has a negative cycle.
By \cref{them:min_cycle}, a negative cycle in $G_{\nu}$ can be detected in $O(n)$ time and using $O(\log n)$ space.
\end{proof}

\subsection{Exact solution}
We now describe the method for determining the value $\nu^{\ast}$ of $G$ exactly.
This is done by making various guesses $\nu$ such that $\nu^{\ast}\geq \nu$ and testing for 
negative cycles in the graph $G_{\nu}=(V,E,\Weight_{\nu})$.
We first determine whether $\nu^{\ast} = 0$, using \cref{lem:ratio_decision}.
In the remaining of this section we assume that $\nu^{\ast}\neq 0$.

\noindent{\bf Solution overview.}
Consider that $\nu^{\ast}> 0$.
First, we either find that $\nu^{\ast}\in(0,1)$ (hence $\lfloor\nu^{\ast}\rfloor=0$),
or perform an \emph{exponential search} of $O(\log\nu^{\ast} )$ iterations
to determine $j\in \Nats^+$ such that $\nu^{\ast}\in[2^{j-1}, 2^{j}]$.
In the latter case, we perform a binary search of $O(\log \nu^{\ast})$ iterations
in the interval $[2^{j-1}, 2^{j}]$ to determine $\lfloor \nu^{\ast}\rfloor$
(see \cref{fig:search}).
Then, we can write $\nu^{\ast} = \lfloor \nu^{\ast}\rfloor + x$, 
where $x<1$ is an irreducible fraction $\frac{a'}{b}$.
It has been shown~\cite{Papadimitriou79} that such $x$ can be
determined by evaluating $O(\log b)$ inequalities of the form $x\geq\nu$.
The case for $\nu^{\ast}<0$ is handled similarly.

\begin{lemma}\label{lem:integer_part}
Let $\nu^{\ast}\neq 0$ be the ratio value of $G$. The value $\lfloor \nu^{\ast} \rfloor$ can be 
obtained by evaluating $O( \log  |\nu^{\ast}|)$ inequalities
of the form $\nu^{\ast}\geq \nu$.
\end{lemma}
\begin{proof}
First determine whether $\nu^{\ast}> 0$, and assume w.l.o.g. that this is the case (the process is similar if $\nu^{\ast}<0$).
Perform an exponential search on the interval $(0,2\cdot \lfloor\nu^{\ast}\rfloor)$
by a sequence of evaluations of the inequality $\nu^{\ast}\geq \nu_i=2^i$.
After $\log \lfloor \nu^{\ast}\rfloor+1$ steps we either have $\lfloor\nu^{\ast}\rfloor \in (0,1)$,
or have determined a $j>0$ such that $\nu^{\ast}\in [\nu_{j-1}, \nu_j]$.
Then, perform a binary search in the interval $[\nu_{j-1}, \nu_j]$,
until the running interval $[\ell, r]$ has length at most $1$. 
Since $\nu_j - \nu_{j-1} = \nu_{j-1} \leq \nu^{\ast}$, 
this will happen after at most $\log \lceil \nu^{\ast} \rceil $ steps.
Then either $\lfloor \nu^{\ast} \rfloor = \lfloor \ell \rfloor $ or 
$\lfloor \nu^{\ast} \rfloor = \lfloor r \rfloor $, which can be determined by
evaluating the inequality $ \nu^{\ast} \geq \lfloor r \rfloor$.
A similar process can be carried out when $\nu^{\ast}<0$.
\cref{fig:search} shows an illustration of the search.
\end{proof}

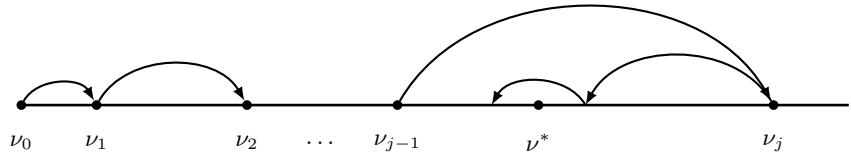
\begin{figure}[!h]
\centering
\begin{tikzpicture}[thick, >=latex,
pre/.style={<-,shorten >= 1pt, shorten <=1pt, thick},
post/.style={->,shorten >= 1pt, shorten <=1pt,  thick},
und/.style={very thick, draw=gray},
node/.style={circle, minimum size=1mm, draw=black!100, fill=black!100, inner sep=1},
rootbag/.style={ellipse, minimum height=7mm,minimum width=14mm,draw=black!80, line width=2.5pt},
virt/.style={circle,draw=black!50,fill=black!20, opacity=0}]

\def \ylabel{-0.5}
\def \step{1}
\def \bend{60}

\draw [draw=black, line width =1] (0,0) to (11,0);

\node	[node] (mu0)	at	(0*\step,0)	{};
\node	[]	at	(0,\ylabel)	{$\nu_0$};

\node	[node] (mu1)	at	(1*\step,0)	{};
\node	[]	at	(1*\step,\ylabel)	{$\nu_1$};

\node	[node] (mu2)	at	(3*\step,0)	{};
\node	[]	at	(3*\step,\ylabel)	{$\nu_2$};

\node	[]	at	(4*\step,\ylabel)	{$\dots$};

\node	[node] (mui_1)	at	(5*\step,0)	{};
\node	[]	at	(5*\step,\ylabel)	{$\nu_{j-1}$};

\node	[node] (mui)	at	(10*\step,0)	{};
\node	[]	at	(10*\step,\ylabel)	{$\nu_j$};

\draw [draw=black, line width =0.8, ->, bend left=\bend] (mu0) to (mu1);
\draw [draw=black, line width =0.8, ->, bend left=\bend] (mu1) to (mu2);
\draw [draw=black, line width =0.8, ->, bend left=\bend] (mui_1) to (mui);
\draw [draw=black, line width =0.8, ->, bend right=\bend] (mui) to ($(mui_1)+0.5*(mui)-0.5*(mui_1)$);
\draw [draw=black, line width =0.8, ->, bend right=\bend] ($(mui_1)+0.5*(mui)-0.5*(mui_1)$) to ($(mui_1)+0.25*(mui)-0.25*(mui_1)$);

\node	[node] (mu)	at	($(mui_1)+0.375*(mui)-0.375*(mui_1)$)	{};
\node	[]	at	($(mui_1)+0.375*(mui)-0.375*(mui_1)+(0,\ylabel)$)	{$\nu^{\ast}$};

\end{tikzpicture}
\caption{Exponential search followed by a binary search to determine $\lfloor \nu^{\ast}\rfloor$}\label{fig:search}
\end{figure}

Let $T_{\max}=\max_e\Timef(e)$ be the largest weight of an edge wrt $\Timef$.
Since $\nu^{\ast}$ is a number with denominator at most $(n-1)\cdot T_{\max}$, it can be determined
exactly by carrying the binary search of \cref{lem:integer_part} until the length
of the running interval becomes at most $\frac{1}{((n-1)\cdot T_{\max})^2}$ (thus containing a unique rational
with denominator at most $(n-1)\cdot T_{\max}$). 
Then $\nu^{\ast}$ can be obtained by using continued fractions, e.g. as in~\cite{Kwek03}.
We rely in the work of Papadimitriou~\cite{Papadimitriou79} to obtain a tighter bound.

\begin{lemma}\label{lem:fractional_part}
Let $\nu^{\ast}\neq 0$ be the ratio value of $G$, such that $\nu^{\ast}$ is the irreducible fraction $\frac{a}{b}\in (-1,1)$. 
Then $\nu^{\ast}$ can be determined by evaluating $O(\log b)$ inequalities of the form $\nu^{\ast}\geq \nu$.
\end{lemma}
\begin{proof}
Consider that $\nu^{\ast} > 0$ (the proof is similar when $\nu^{\ast} < 0$).
It is shown in~\cite{Papadimitriou79} that a rational with denominator at most $b$
can be determined by evaluating $O(\log b)$ inequalities of the form $\nu^{\ast}\geq \nu$.
We remark that $b$ is not required to be known,  although the work of~\cite{Papadimitriou79} assumes
that a bound on the denominator of $\nu^{\ast}$ is known in advance.
\end{proof}

\begin{theorem}\label{them:min_ratio_exact}
Let $G=(V,E,\Weight, \Timef)$ be a weighted graph of $n$ nodes with constant treewidth,
and $\lambda=\max_u |a_u\cdot b_u|$ such that $\nu^{\ast}(u)$ is the irreducible fraction $\frac{a_u}{b_u}$.
Let $\Time(G)$ and $\Space(G)$ denote the required time and space for constructing a balanced binary tree-decomposition
$\Tree(G)$ of $G$ with constant width.
The minimum ratio cycle problem for $G$ can be computed in
\begin{compactenum}
\item\label{item:thm_ratio_time_space} $O(\Time(G)+n\cdot \log (\lambda))$ time and $O(\Space(G)+n)$ space; and
\item\label{item:thm_ratio_space} $O(\Space(G)+\log n)$ space.
\end{compactenum}
\end{theorem}
\begin{proof}
In view of \cref{rem:scc} the graph $G$ is strongly connected and has a minimum ratio value $\nu^{\ast}$.
Let $\nu^{\ast}=\lfloor \nu^{\ast} \rfloor + \frac{a'}{b}$ with $|\frac{a'}{b}|<1$.
By \cref{lem:integer_part}, $\lfloor \nu^{\ast} \rfloor$ can be determined by evaluating $O(\log |\nu^{\ast}|)=O(\log |a|)$ inequalities
of the form $\nu^{\ast} \geq \nu$, and by \cref{lem:fractional_part}, $\frac{a'}{b}$ can be determined by evaluating
$O(b)$ such inequalities.
A balanced binary tree-decomposition $\Tree(G)$ can be constructed once
in $\Time(G)$ time and $\Space(G)$ space, and stored in $O(n)$ space.
$\Tree(G)$ is also a tree-decomposition of every $G_{\nu}$ required by Claim~\ref{claim:value}.
By \cref{them:min_cycle}
a negative cycle in $G_{\nu}$ can be detected in $O(n)$ time and using $O(\log n)$ space.
This concludes \cref{item:thm_ratio_time_space}. 
Item~\ref{item:thm_ratio_space} is obtained by the same process, but with re-computing $\Tree(G)$
every time $\Cyclealgo$ traverses from a bag to a neighbor (thus not storing $\Tree(G)$ explicitly).
\end{proof}

Using \cref{them:tree_decomp} we obtain from \cref{them:min_ratio_exact} the following corollary.

\begin{corollary}\label{cor:min_ratio_exact}
Let $G=(V,E,\Weight, \Timef)$ be a weighted graph of $n$ nodes with constant treewidth,
and $\lambda=\max_u |a_u\cdot b_u|$ such that $\nu^{\ast}(u)$ is the irreducible fraction $\frac{a_u}{b_u}$.
The minimum ratio value problem for $G$ can be computed in
\begin{compactenum}
\item$O(n\cdot \log (\lambda))$ time and $O(n)$ space; and
\item $O(\log n)$ space.
\end{compactenum}
\end{corollary}

By setting $\Timef(e)=1$ for each $e\in E$ in \cref{cor:min_ratio_exact} we obtain the following corollary
for the minimum mean cycle.

\begin{corollary}\label{cor:min_mean_exact}
Let $G=(V,E,\Weight)$ be a weighted graph of $n$ nodes with constant treewidth,
and $\lambda = \max_u |\mu^{\ast}(u)|$.
The minimum mean value problem for $G$ can be computed in
\begin{compactenum}
\item $O(n\cdot \log (\lambda))$ time and $O(n)$ space; and
\item $O(\log n)$ space.
\end{compactenum}
\end{corollary}

\subsection{Approximating the minimum mean cycle}
We now focus on the minimum mean cycle problem, and present algorithms
for $\epsilon$-approximating the mean value $\mu^{\ast}$ of $G$ 
for any $0<\epsilon<1$ in $O(n\cdot \log (n/\epsilon))$ time, i.e., independent of $\mu^{\ast}$.

\noindent{\bf Approximate solution in the absence of negative cycles.}
We first consider graphs $G$ that do not have negative cycles. 
Let $C$ be a minimum mean value cycle, and $C'$ a minimum weight simple cycle in $G$,
and note that $\mu^{\ast}\in [0,\Weight(C')]$.
Additionally, we have

\[
\Weight(C')\leq \Weight(C) \implies \Weight(C')\leq \frac{n}{|C|}\cdot \Weight(C)\implies \Weight(C')\leq (n)\cdot \mu^{\ast}
\]

Consider a binary search in the interval $[0,\Weight(C')]$, 
which in step $i$ approximates $\mu^{\ast}$ by the right endpoint $\mu_i$ of its current interval.
The error is bounded by the length of the interval, hence $\mu_i-\mu^{\ast}\leq\Weight(C')\cdot 2^{-i} \leq (n-1)\cdot \mu^{\ast}\cdot 2^{-i}$.
To approximate within a factor $\epsilon$ we require 

\begin{equation}\label{eq:approximation}
2^{-i}\cdot (n-1)\leq \epsilon \implies i\geq \log(n)+\log(1/\epsilon)
\end{equation}

steps.

\begin{remark}\label{rem:no_approx}
Note that for the minimum ratio value we have  $\Weight(C')\leq W'\cdot n\cdot \nu^{\ast}$, 
where $W'=\max_{e\in E} \Timef(e)$.
For $\epsilon$-approximating $\nu^{\ast}$ we would need $i\geq \log(n\cdot W'/\epsilon)$ steps.
\end{remark}

\noindent{\bf Approximate solution in the presence of negative cycles.}
We now turn our attention to $\epsilon$-approximating $\mu^{\ast}$ in the presence of negative cycles in $G$.
Note that uniformly increasing the weight of each edge so that no negative edges exist
does not suffice, as the error can be of order $\epsilon\cdot |W^{-}|$ rather than $\epsilon\cdot \mu^{\ast}$,
where $W^{-}$ is the minimum edge weight.

Instead, let $c$ be the value returned by $\Cyclealgo$ on input $G$.
Item~\ref{item:them_neg_cycle1} of \cref{them:min_cycle}  guarantees that for the weight function
$\Weight_{-|c|}(e)=\Weight(e)+|c|$, the graph $G_{-|c|}=(V,E,\Weight_{-|c|})$ has no negative cycles
(although it might still have negative edges).
The following lemma states that $\mu^{\ast}$ can be $\epsilon$-approximated
by $\epsilon'$-approximating the value $\mu'^{\ast}$ of $G_{-|c|}$, 
for some $\epsilon'$ polynomially (in $n$) smaller than $\epsilon$.

\begin{lemma}\label{lem:approximation}
Let $\mu^{\ast}$ and $\mu'^{\ast}$ be the value of $G$ and $G_{-|c|}$ respectively, and $\epsilon$ 
some desired approximation factor of $\mu^{\ast}$, with $0<\epsilon<1$.
There exists an $\epsilon'=\epsilon/n^{O(1)}$ such that if $\mu'$ is an $\epsilon'$-approximation of $\mu'^{\ast}$ in $G_{-|c|}$,
then $\mu=\mu'-|c|$ is an $\epsilon$-approximation of $\mu^{\ast}$ in $G$.
\end{lemma}
\begin{proof}
By construction, we have $\mu'^{\ast}=\mu^{\ast}+|c|$, where $c$ defined above is the
value returned by $\Cyclealgo$ on $G$. 
Let $c^{\ast}$ be the weight of a minimum-weight simple cycle in $G$. 
By~\cref{them:min_cycle} \cref{item:them_neg_cycle2}, we have that $|c| = |c^{\ast}|\cdot n^{O(1)}$.
Note that $|c^{\ast}|\leq (n-1)\cdot |\mu^{\ast}|$, hence $\mu'^{\ast}= \mu^{\ast}+|c^{\ast}|\cdot n^{O(1)}\leq |\mu^{\ast}|\cdot \alpha$
for $\alpha=n^{O(1)}$.
Let $\epsilon'=\epsilon/\alpha$.
By $\epsilon'$-approximating $\mu'^{\ast}$ by $\mu'$ we have

\[
|\mu'-\mu'^{\ast}|\leq \epsilon'\cdot |\mu'^{\ast}| \implies |(\mu'-|c|)-(\mu'^{\ast}-|c|)|\leq \epsilon'\cdot |\mu'^{\ast}| \implies |\mu-\mu^{\ast}|\leq \epsilon'\cdot |\mu^{\ast}|\cdot  \alpha  \leq \epsilon\cdot |\mu^{\ast}|
\]

The desired result follows.
\end{proof}

\begin{theorem}\label{them:min_mean_approx}
Let $G=(V,E,\Weight)$ be a weighted graph of $n$ nodes with constant treewidth.
For any $0<\epsilon<1$, the minimum mean value problem can be $\epsilon$-approximated in
$O(n\cdot \log(n/\epsilon))$ time and $O(n)$ space.
\end{theorem}
\begin{proof}
In view of \cref{rem:scc} the graph $G$ is strongly connected and has a minimum mean value $\mu^{\ast}$.
First, we construct a balanced binary tree-decomposition $\Tree(G)$ of $G$ in $O(n\cdot \log n)$
time and $O(n)$ space \cref{them:tree_decomp}.
Let $c$ be the value returned by $\Cyclealgo$ on the input graph $G$.
If $c\geq 0$, by \cref{lem:cyclealgo_correctness} we have $\mu^{\ast}\geq 0$, and by \cref{eq:approximation} $\mu^{\ast}$ can be $\epsilon$-approximated
in $O(\log(n/\epsilon))$ steps.
If $c<0$, we construct the graph $G_{-|c|}=(V,E,\Weight_{-|c|})$.
By \cref{lem:approximation}, $\mu^{\ast}$ can be $\epsilon$-approximated by $\epsilon'$ approximating
the mean value $\mu'^{\ast}$ of $G_{-|c|}$, where $\epsilon'=\frac{\epsilon}{n^{O(1)}}$. By construction, $G_{-|c|}$ does not contain negative cycles,
thus $\mu'^{\ast}\geq 0$, and by \cref{eq:approximation} $\mu'^{\ast}$ can be approximated in
$O(\log(n/\epsilon'))=O(\log(n/\epsilon))$ steps.
By \cref{lem:cyclealgo_complexity}, each step requires $O(n)$ time.
The statement follows.
\end{proof}

%% file: initial_credit.tex
\section{The Minimum Initial Credit Problem}\label{sec:init_credit}
In the current section we present algorithms for solving the minimum initial credit
problem on weighted graphs $G=(V,E,\Weight)$.
We first deal with arbitrary graphs, and provide (i)~an $O(n\cdot m)$ algorithm
for the decision problem, and (ii)~ an  $O(n^2\cdot m)$ for the value problem,
improving the previously best upper bounds.
Afterwards we adapt our approach on graphs of constant treewidth  
to obtain an $O(n\cdot \log n)$ algorithm for the value problem.

\noindent{\bf Non-positive minimum initial credit.}
For technical convenience we focus on a variant of the minimum initial credit problem,
where energies are non-positive, and the goal is to keep partial sums of path prefixes non-positive.
Formally, given a weighted graph $G=(V,E,\Weight)$, the non-positive minimum initial credit value problem
asks to determine for each node $u\in V$ the largest energy value $\Energy(u)\in \Intsminus$
with the following property: there exists an infinite path $\mathcal{P}=(u_1, u_2\dots)$ with $u=u_1$,
such that for every finite prefix $P$ of $\mathcal{P}$ we have $\Energy(u)+\Weight(P)\leq 0$.
Conventionally, we let $\Energy(u)=-\infty$ if no finite such value exists.
The associated decision problem asks given a node $u$ and an initial credit $c\in \Ints_{\leq 0}$
whether $\Energy(u)\geq c$.
Hence, here minimality is wrt the absolute value of the energy.
A solution to the standard minimum initial credit problem can be obtained by
inverting the sign of each edge weight and solving the non-positive minimum initial credit problem in the resulting graph.

We start with some definitions and claims that will give the intuition for the algorithms to follow.
First, we define the minimum initial credit of a pair of nodes $u,v$, 
which is the energy to reach $v$ from $u$ (i.e., the energy is wrt a finite path).

\noindent{\bf Finite minimum initial credit.}
For nodes $u,v\in V$, we denote with $\Energy_v(u)\in \Intsminus$
the largest value with the following property:
there exists a path $P:u\Path v$
such that for every prefix $P'$ of $P$ we have $\Energy_v(u)+\Weight(P')\leq 0$.
Note that for every pair of nodes $u,v\in V$, we have $\Energy(u)\geq \Energy_v(u)+\Energy(v)$.
Conventionally, we let $\Energy_v(u)=-\infty$ if no such value exists
(i.e., there is no path $u\rightsquigarrow v$).

\begin{remark}\label{rem:energy_distance}
For any $u\in V$, let $P:u\Path v$ be a witness path for $\Energy_v(u)>-\infty$. Then
\[
\Energy_v(u)+\Weight(P)\leq 0 \implies \Energy_v(u)\leq -\Weight(P) \leq -\Distance(u,v)
\]
i.e., the energy to reach $v$ from $u$ is upper bounded by minus the distance from $u$ to $v$.
\end{remark}

\noindent{\bf Highest-energy nodes.}
Given a (possibly infinite) path $P$ with $\Weight(P)<\infty$, 
we say that a node $x\in P$ is a \emph{highest-energy node} of $P$
if there exists a \emph{highest-energy prefix} $P_1$ of $P$ ending in $x$ such that for any prefix $P_2$ of $P$
we have $\Weight(P_1)\geq \Weight(P_2)$.
Note that since the weights are integers,
for every pair of paths $P'_1$, $P'_2$, it is either $|\Weight(P'_1)-\Weight(P'_2)|=0$ 
or $|\Weight(P'_1)-\Weight(P'_2)|\geq 1$.
Therefore the set $\{\Weight(P_i)\}_i$ of weights of prefixes of  $P$
has a maximum, and thus a highest-energy node always exists when $\Weight(P)<\infty$.
The following properties are easy to verify:

\begin{compactenum}
\item If $x$ is a highest-energy node in a path $P:u\Path v$,
then $\Energy_v(x)=0$.
\item If $x$ is a highest-energy node in an infinite path $\mathcal{P}$, then $\Energy(x)=0$.
\end{compactenum}

The following claim states that the energy $\Energy(u)$
of a node $u$ is the maximum energy $\Energy_v(u)$ to reach a $0$-energy node $v$.
\begin{claim}\label{claim:zero_energy}
For every $u\in V$, we have $\Energy(u)=\max_{v:\Energy(v)=0}\Energy_v(u)$.
\end{claim}
\begin{proof}
The direction $\Energy(u)\geq \max_{v:\Energy(v)=0}\Energy_v(u)$ is straightforward.
For the other direction, consider that $\Energy(u)>-\infty$
(trivially, $-\infty\leq \max_{v:\Energy(v)=0}\Energy_v(u)$)
and let $\mathcal{P}$ be a witness path for $\Energy(u)$.
Since $\Energy(u)>-\infty$, we have $\Weight(\mathcal{P})<\infty$, and
$\mathcal{P}$ has some highest-energy node $x$, thus $\Energy(x)=0$.
Since $x$ is on the witness $\mathcal{P}$ of $\Energy(u)$,
we have $\Energy(u)\leq \Energy_x(u)\leq \max_{v:\Energy(v)=0}\Energy_v(u)$.
The result follows.
\end{proof}

\subsection{The decision problem for general graphs}
Here we address the decision problem, namely, given some node $u\in V$
and an initial credit $c\in \Ints_{\leq 0}$, determine whether $\Energy(u)\geq c$.
The following claim states that if $\Energy(u)\geq c$, then
a non-positive cycle can be reached from $u$ with initial credit $c$,
by paths of length less than $n$.

\begin{claim}\label{claim:energy_decision}
For every $u\in V$ and $c\in \Ints_{\leq 0}$, we have that $\Energy(u) \geq c$ iff 
there exists a simple cycle $C$ such that (i)~$\Weight(C) \leq 0$ and
(ii)~for every $v\in C$ we have that $\Energy_v(u) \geq c$, which is 
witnessed by a path $P_v:u\Path v$ with $|P_v|<n$.
\end{claim}
\begin{proof}
For the one direction, if $\Weight(C)\leq 0$ we have $\Weight(C^{\omega})< \infty$, thus $C$
contains a $0$-energy node $w$. By Claim~\ref{claim:zero_energy}, $\Energy(u)= \max_{v:\Energy(v)=0}\Energy_v(u) \geq \Energy_w(u)\geq c$.
For the other direction, let $\mathcal{P}$ be a witness path for $\Energy(u)$, and we can assume w.l.o.g. that
$\mathcal{P}$ does not contain positive cycles. Then for every prefix $P_v:u\Path v$ of $\mathcal{P}$ we have
$\Energy(u) + \Weight(P_v)\leq 0$, thus $\Energy_v(u) \geq \Energy(u) \geq c$,
and the $n$-th such prefix contains a non-positive cycle $C$. The result follows.
\end{proof}

\noindent{\bf Algorithm $\Decisionenergyalgo$.}
Claim~\ref{claim:energy_decision} suggests a way to decide whether $\Energy(u) \geq c$.
First, we start with energy $c$ from $u$, and perform a sequence of $n-1$ relaxation steps,
similar to the Bellman-Ford algorithm, to discover the set $V_u^c$ of nodes that can be reached
from $u$ with initial credit $c$ by a path of length at most $n-1$. Afterwards, we perform a Bellman-Ford computation on
the subgraph $G\restr V_u^c$ induced by the set $V_u^c$. 
By Claim~\ref{claim:energy_decision}, we have that $\Energy(u) \geq c$ iff $G\restr V_u^c$
contains a non-positive cycle. Algorithm~\ref{algo:decisionenergyalgo} ($\Decisionenergyalgo$) gives a formal description.
The \emph{for} loop in lines~\ref{line:relaxation_loop_begin}-\ref{line:relaxation_loop_end} is similar to 
the procedure ROUND from the algorithm of~\cite{Bouyer08}.

\noindent{\bf Detecting non-positive cycles.}
It is known that the Bellman-Ford algorithm can detect negative cycles.
To detect non-positive cycles in a graph $G$ with $n$ nodes and weight function $\Weight$,
we execute Bellman-Ford on $G$
with a slightly modified weight function $\Weight'$ for which $\Weight'(e)=\Weight(e)-\frac{1}{n}$.
Then for any simple cycle $C$ in $G$ we have $\Weight(C)\leq 0$ iff $\Weight'(C)<0$. 
Indeed,
\[
\Weight'(C)<0 \iff \sum_{e\in C} \Weight(e) -\sum_{e\in C} \frac{1}{n}<0
\iff \Weight(C)<\frac{|C|}{n} \iff \Weight(C)\leq 0
\]

since $|C|\leq n$ and $\Weight(C)\in \Ints$.

\begin{algorithm}
\small
\DontPrintSemicolon
\caption{$\Decisionenergyalgo$}\label{algo:decisionenergyalgo}
\KwIn{A weighted graph $G=(V,E,\Weight)$, a node $u\in V$, an initial energy $c\in \Ints_{\leq 0}$}
\KwOut{$\True$ iff $\Energy(u)\geq c$}
\BlankLine
\tcp{Initialization}
\ForEach{$v \in V$}{
Assign $D(s)\leftarrow \infty$\\
}
Assign $D(u)\leftarrow c$\\
Assign $V_u^c\leftarrow \{u\}$\\
\tcp{$n-1$ relaxation steps to discover $V_u^c$}
\For{$i\leftarrow 1$ \KwTo $n-1$}{\label{line:relaxation_loop_begin}
\ForEach{$(v,w) \in E$}{
\uIf{$D(w)\geq D(v) + \Weight(v,w)$ and $D(v) + \Weight(v,w) \leq 0$ }{
Assign $D(w)\leftarrow D(v) + \Weight(v,w)$\\
Assign $V_u^c \leftarrow V_u^c\cup \{w\}$\\
}
}
}
\label{line:relaxation_loop_end}
Execute Bellman-Ford on $G\restr V_u^c$\\
\Return $\True$ iff a non-positive cycle is discovered
\end{algorithm}

The correctness of $\Decisionenergyalgo$ follows directly from Claim~\ref{claim:energy_decision}.
The time complexity is $O(n\cdot m)$ time spent in the \emph{for} loop of lines \ref{line:relaxation_loop_begin}-\ref{line:relaxation_loop_end},
plus $O(n\cdot m)$ time for the Bellman-Ford. We thus obtain the following theorem.

\begin{theorem}\label{them:energy_decision}
Let $G=(V,E,\Weight)$ be a weighted graph of $n$ nodes and $m$ edges.
Let $u\in V$ be an initial node, and $c\in \Ints_{\leq 0}$ be an initial credit.
The decision problem of whether $\Energy(u)\geq c$ can be solved in $O(n\cdot m)$ time
and $O(n)$ space.
\end{theorem}

\subsection{The value problem for general graphs}
We now turn our attention to the value version of the minimum initial credit problem,
where the task is to determine $\Energy(u)$ for every node $u$.
The following claim establishes that if for all energies to reach some node $v$ we have $\Energy_v(w)<0$,
then $\Energy_v(u)=-\Distance(u,v)$, i.e., the energy to reach $v$ from every node $u$ is minus the distance
from $u$ to $v$.

\begin{claim}\label{claim:energy_distance}
If for all $w\in V\setminus\{v\}$ we have $\Energy_v(w)<0$, then for each $u\in V\setminus\{v\}$ we have $\Energy_v(u)=-\Distance(u,v)$.
\end{claim}
\begin{proof}
Let $P:u\Path v$ be a witness path to the distance, i.e., $\Weight(P)=\Distance(u,v)<\infty$ 
(if $\Distance(u,v)=\infty$ the statement is trivially true).
Since every highest-energy node $x$ of $P$ has $\Energy_v(x)=0$, we have that $x=v$.
Hence, $P$ is a highest-energy prefix of itself, and for each prefix $P'$ of $P$ we have
$ -\Weight(P) + \Weight(P') \leq 0$  and thus $\Energy_v(u)\geq -\Weight(P)= -\Distance(u,v)$.
By \cref{rem:energy_distance}, it is $\Energy_v(u)\leq -\Distance(u,v)$. The result follows.
\end{proof}

\noindent{\bf An $O(n^2\cdot m)$ time solution to the value problem.}
Claim~\ref{claim:energy_distance} together with \cref{them:energy_decision} lead to an 
$O(n^2\cdot m)$ method for solving the minimum initial credit value problem. 
First, we compute the set $X=\{v\in V: \Energy(v)=0\}$ in $O(n^2\cdot m)$
time, by testing whether $\Energy(u)\geq 0$ for each node $u$. 
Afterwards, we contract the set $X$ to a new node $z$, and by Claim~\ref{claim:zero_energy}
for every remaining node $u$ we have  $\Energy(u)=\max_{v\in X}\Energy_v(u)=\Energy_z(u)$.
Since $u\not \in X$, the energy of $u$ is strictly negative, and thus $\Energy_z(u)<0$.
Finally, by Claim~\ref{claim:energy_distance}, we have $\Energy_z(u)=-\Distance(u,z)$.
Hence it suffices to compute the distance of each node $u$ to $z$, which can be obtained
in $O(n\cdot m)$ time. 

In the remaining of this subsection we provide a refined solution of $O(k\cdot n\cdot m)$ time,
where $k=|X|+1$ is the number of $0$-energy nodes (plus one).
Hence this solution is faster in graphs where $k=o(n)$.
This is achieved by algorithm $\Energynodesalgo$ for computing the set
$X$ faster.

\noindent{\bf Determining the $0$-energy nodes.}
The first step for solving the minimum initial credit problem
is determining the set $X$ of all $0$-energy nodes of $G$.
To achieve this, we construct the graph $G_2=(V_2, E_2, \Weight_2)$
with a fresh node $z\not\in V$ as follows: 

\begin{compactenum}
\item The node set is $V_2=V\cup\{z\}$,
\item The edge set is $E_2=E\cup  (\{z\}\times V)$,
\item The weight function $\Weight_2:E_2\rightarrow \Ints$ is

\[
\Weight_2(u,v) =
\left\{
	\begin{array}{ll}
		0  & \mbox{if } u=z \\
		\Weight(u,v) & \mbox{otherwise}
	\end{array}
\right.
\]
\end{compactenum}

\begin{remark}\label{rem:highest_energy_node}
Since for every outgoing edge $(z,x)$ of $z$ we have $\Weight_2(z,x)=0$,
if $z$ is a highest-energy node in a path of $G_2$, so is $x$.
Hence every non-positive cycle in $G_2$ has a highest-energy node other than $z$.
\end{remark}

Note that for every $u\in V$, the energy $\Energy(u)$ is the same in $G$ and $G_2$.

\noindent{\bf Algorithm $\Energynodesalgo$.}
Algorithm~\ref{algo:energynodesalgo} describes $\Energynodesalgo$ for obtaining the
set of all $0$-energy nodes in $G_2$.
Informally, the algorithm performs a sequence of modifications on a graph $\Graphstruct$,
initially identical to $G_2$.
In each step, the algorithm executes a Bellman-Ford computation on the current graph $\Graphstruct$ with $z$ as the source node,
as long as a non-positive cycle $C$ is discovered. 
For every such $C$, it determines a highest-energy node $w$
of $C$, and modifies $\Graphstruct$ by replacing every
incoming edge $(x,w)$ with an edge $(x,z)$ of the same weight,
and then removing $w$. 
See \cref{fig:credit} for an illustration.

\begin{algorithm}
\small
\DontPrintSemicolon
\caption{$\Energynodesalgo$}\label{algo:energynodesalgo}
\KwIn{A weighted graph $G_2=(V_2,E_2,\Weight_2)$}
\KwOut{The set $\{v\in V_2\setminus\{z\}: \Energy(v)=0\}$}
\BlankLine
Initialize sets $\Nodeset\leftarrow V_2$, $\Edgeset\leftarrow E_2$ and map $\Weightmap\leftarrow \Weight_2$\\
Let $\Graphstruct=(\Nodeset, \Edgeset, \Weightmap)$\\
Initialize set $X\leftarrow \emptyset$\\
\While{$\True$}{\label{line:while_true}
Execute Bellman-Ford from source node $z$ in $\Graphstruct$\label{line:bellman_ford}\\
\eIf{exists non-positive cycle $C$}{
Determine a highest-energy node $w\neq z$ in $C$\label{line:zero_energy}\\
Assign $X\leftarrow X\cup\{w\}$\\
\ForEach{edge $(x,w)\in \Edgeset$}{\label{line:for_each_edge}
\eIf{$(x,z)\not\in\Edgeset$}{
Assign $\Edgeset\leftarrow \Edgeset \cup \{(x,z)\}$\label{line:mod1}\\
Assign $\Weightmap(x,z)\leftarrow \Weight_2(x,w)$
}
{
Assign $\Weightmap(x,z)\leftarrow \min(\Weight_2(x,w),\Weightmap(x,z))$\label{line:mod2}\\
}
}
Assign $\Nodeset\leftarrow \Nodeset\setminus\{w\}$\\
}
{\Return{$X$}}
}
\end{algorithm}

As $0$-energy nodes are discovered,  $\Energynodesalgo$ performs a sequence of modifications to the 
graph $\Graphstruct$.
We denote with $\Graphstruct^k$ the graph $\Graphstruct$ after the $k$-th node
has been added to $X$ (and $\Graphstruct^0=G_2$).
We also use the superscript-$k$ in our graph notation to make it specific to $\Graphstruct^k$
(e.g. $\Distance^k(u,z)$ and $\Energy_z^k(u)$ denote respectively the distance from $u$ to $z$, and the
energy to reach $z$ from $u$ in $\Graphstruct^k$).
The following two lemmas establish the correctness of $\Energynodesalgo$.

\begin{lemma}\label{lem:energynodesalgo_cor1}
For every $w\in X$ we have $\Energy(w)=0$.
\end{lemma}
\begin{proof}
The proof is by induction on the size of $X$. It is trivially true when $|X|=0$.
For the inductive step, let $w$ be the $k+1$-th node added in $X$.
By line~\ref{line:zero_energy}, $w$ is a highest-energy node in a non-positive cycle $C$ of $\Graphstruct^{k}$.
We split into two cases.
\begin{compactenum}
\item If $z\not\in C$, then $C$ is also a cycle of $G$, hence $w$ is a highest-energy node in the infinite path $\mathcal{P}=C^\omega$ of $G$,
and $\Energy(w)=0$.
\item If $z\in C$, let $x$ be the node before $z$ in $C$. 
By the modifications of lines~\ref{line:mod1} and \ref{line:mod2},
it is $\Weightmap^k(x,z)=\Weight_2(x,w')$, where $w'$ is a node that has been added to $X$
when the algorithm run on $\Graphstruct^i$ for some $i< k$.
It follows that $w$ is a highest-energy node in a path $P:z\Path w'$ in $G_2$,
and thus a highest-energy node in a suffix $P':w\Path w'$ of $P$,
where $P'$ is a path in $G$. Hence $\Energy_{w'}(w)=0$.
By the induction hypothesis, $w'$ is a $0$-energy node, i.e., $\Energy(w')=0$,
thus by Claim~\ref{claim:zero_energy} we have $\Energy(w)\geq \Energy_{w'}(w)=0$.
\end{compactenum}
The result follows.
\end{proof}

\begin{lemma}\label{lem:energynodesalgo_cor2}
For every $w\in V: \Energy(w)=0$ we have $w\in X$.
\end{lemma}
\begin{proof}
Consider any $w\in V: \Energy(w)=0$. For some $i\in \Nats$, we say that
$\Graphstruct^i$ ``is aware of $w$'' if either $\Graphstruct^i$
has a non-positive cycle $C:w\Path w$, or $w\in X$ when $|X|=i$.
Note that when $\Energynodesalgo$ terminates there are no non-positive cycles in $\Graphstruct^{|X|}$.
Hence, it suffices to argue that there exists a $k\in\Nats$ such that
for each $i\geq k$, $\Graphstruct^i$ is aware of $w$.
We first argue that there exists a $k$ for which $\Graphstruct^k$ is aware of $w$.

Let $\mathcal{P}$ be a witness for $\Energy(w)=0$,
hence $\mathcal{P}$ traverses a non-positive cycle $C_1$ in $G$, 
thus $C_1$ exists in $\Graphstruct^0$.
Then there exists a smallest $j\in \Nats$ such that 
some node $w'$ of $\mathcal{P}$ is identified as a highest-energy node
in a non-positive cycle $C_2$ (possibly $C_1=C_2$), and inserted to $X$.
If $w=w'$, we have that $\Graphstruct^j$ is aware of $w$.
Otherwise, since $\Energy(w)=0$ and $w'$ is a node in the witness $\mathcal{P}$,
we have $\Energy_{w'}(w)=0$.
By the choice of $w'$, the path $\mathcal{P}$ exists in $\Graphstruct^j$,
therefore $\Energy^j_{w'}(w)=\Energy_{w'}(w)=0$,
and by \cref{rem:energy_distance}, we have $\Distance^j(w,w')\leq 0$.
It is straightforward that after the modifications in lines~\ref{line:mod1} and \ref{line:mod2},
we have that $\Distance^{j+1}(w,z)\leq \Distance^j(w,w')\leq 0$, and since $\Weightmap^j(z,w)=\Weight_2(z,w)=0$,
we have a non-positive cycle $C:w\Path w$ in $\Graphstruct^{j+1}$ through $z$.
Hence either $\Graphstruct^{j}$ or $\Graphstruct^{j+1}$ is aware of $w$,
thus there exists a $k\in\Nats$ for which $\Graphstruct^k$ is aware of $w$.

Finally, observe that the distance $\Distance^i(w,z)$ does not increase in any $\Graphstruct^i$ for $i\geq k$
until $w$ is inserted to $X$, hence for each $i\geq k$, the graph $\Graphstruct^i$ is aware of $w$.
The desired result follows.
\end{proof}

\noindent{\bf Determining the negative-energy nodes.}
Having computed the set $X$ of all the $0$-energy nodes of $G$,
the second step for solving the minimum initial value credit problem
is to determine the energy of every other node $u\in V\setminus X$.
Recall the graph $\Graphstruct^{|X|}=(\Nodeset^{|X|},\Edgeset^{|X|},\Weightmap^{|X|})$ after the end of $\Energynodesalgo$.

\begin{lemma}\label{lem:distance_z}
For every $u\in V\setminus X$ we have $\Energy(u)=-\Distance^{|X|}(u,z)$.
\end{lemma}
\begin{proof}
Consider any node $u\in V\setminus X=\Nodeset^{|X|}\setminus\{z\}$.
By Claim~\ref{claim:energy_distance}, in the graph $G$ we have $\Energy(u)=\max_{v:\Energy(v)=0}\Energy_v(u)$,
and by the correctness of $\Energynodesalgo$ from
\cref{lem:energynodesalgo_cor1} and \cref{lem:energynodesalgo_cor2} we have $X=\{v:\Energy(v)=0\}$,
thus $\Energy(u)=\max_{v\in X}\Energy_v(u)$.
It is straightforward to verify that at the end of $\Energynodesalgo$,
we have $\max_{v\in X}\Energy_v(u)= \Energy^{|X|}_z(u)$, i.e.,
the maximum energy to reach the set $X$ in $G$ is the energy to reach $z$ in $\Graphstruct^{|X|}$.
For all $v\in \Nodeset^{|X|}\setminus\{z\}$ it is $\Energy^{|X|}_z(v)<0$,
otherwise we would have $\Energy(v)=0$ and thus $v\in X$ and $v\not\in \Nodeset^{|X|}$.
Then by Claim~\ref{claim:energy_distance}, $\Energy^{|X|}_z(u)=-\Distance^{|X|}(u,z)$.
We conclude that $\Energy(u)=-\Distance^{|X|}(u,z)$.
\end{proof}

Hence, to compute the energy $\Energy(u)$ of every node $u\in V\setminus X$,
it suffices to compute its distance to $z$ in $\Graphstruct^{|X|}$.	
This is straightforward by reversing the edges of $\Graphstruct^{|X|}$ and
performing a Bellman-Ford computation with $z$ as the source node.
\cref{fig:credit} illustrates the algorithms on a small example.
We obtain the following theorem.

\input{fig_credit.tex}

\begin{theorem}\label{them:energy}
Let $G=(V,E,\Weight)$ be a weighted graph of $n$ nodes and $m$ edges,
and $k=|\{v\in V: \Energy(v)=0\}|+1$. 
The minimum initial credit value problem for $G$ can be solved in $O(k\cdot n\cdot m)$ time
and $O(n)$ space.
\end{theorem}
\begin{proof}
\cref{lem:energynodesalgo_cor1}, \cref{lem:energynodesalgo_cor2} and \cref{lem:distance_z} establish the correctness,
so it remains to argue about the complexity.
The \emph{while} block of line~\ref{line:while_true} is executed at most once for each
$0$-energy node, hence at most $k$ times.
Inside the block, the execution of Bellman-Ford in line~\ref{line:bellman_ford}
requires $O(n\cdot m)$ time and $O(m)$ space.
Since the Bellman-Ford algorithm uses backpointers to remember
predecessors of nodes in distances, a highest-energy node $w$ 
of a non-positive cycle $C$ in line~\ref{line:zero_energy} can be determined
in $O(n)$.
Finally, the \emph{for} loop of line \ref{line:for_each_edge}
will consider each edge $(x,w)$ at most once,
hence it requires $O(m)$ for all iterations of the \emph{while} loop.
Thus $\Energynodesalgo$ uses $O(k\cdot n\cdot m)$ time and $O(n)$ space in total.
The last execution of Bellman-Ford to determine the energy of
negative-energy nodes does not affect the complexity.
The result follows.
\end{proof}

\begin{corollary}\label{cor:energy}
Let $G=(V,E,\Weight)$ be a weighted graph of $n$ nodes and $m$ edges. 
The minimum initial credit value problem for $G$ can be solved in $O( n^2\cdot m)$ time
and $O(n)$ space.
\end{corollary}

\subsection{The value problem for constant-treewidth graphs}
We now turn our attention to the minimum initial credit value problem
for constant-treewidth graphs $G=(V,E,\Weight)$.
Note that in such graphs $m=O(n)$, thus \cref{them:energy}
gives an $O(n^3)$ time solution as compared to the existing $O(n^4\cdot \log (n\cdot W))$ time solution.
This section shows that we can do significantly better, 
namely reduce the time complexity to $O(n\cdot \log n)$.
This is mainly achieved by algorithm $\Energynodesalgotw$
for computing the set $X$ of $0$-energy nodes fast in constant-treewidth graphs.

\noindent{\bf Extended $\Plus$ and $\Min$ operators.}
Recall the graph $G_2=(V_2,E_2,\Weight_2)$ from the last section.
Given $\Tree(G)$, a balanced and binary tree-decomposition $\Tree(G_2)$
of $G_2$ with width increased by $1$ can be easily constructed by
(i)~inserting $z$ to every bag of $\Tree(G)$, and
(ii)~adding a new root bag that contains only $z$.
Let $\Intstriplets=\Ints\times V\times \Ints$.
For a map $f:V_2\times V_2\rightarrow \Ints$,
define the map $g_{f}:V_2\times V_2\rightarrow \Intstriplets$ as

\[
g_{f}(u,v) =
\left\{
	\begin{array}{ll}
		(f(u,v), u, 0) & ~~\mbox{if } f(u,v)<0 \mbox{ or } v=z\\
		(f(u,v), v, f(u,v))  & ~~\mbox{otherwise } \\
	\end{array}
\right.
\]

and for triplets of elements $\alpha_1=(a_1,b_1, c_1),\alpha_2=(a_2,b_2, c_2)\in \Intstriplets$,
define the operations
\begin{compactenum}
\item $\Min(\alpha_1,\alpha_2)=\alpha_i$ with $i=\arg\min_{j\in \{1,2\}}a_j$
\item $\alpha_1\Plus \alpha_2 = (a_1+a_2, b, c)$, where $c=\max(c_1,a_1+c_2)$ and $b=b_1$ if $c=c_1$ else $b=b_2$.
\end{compactenum}

In words, if $f$ is a weight function, then $g_{f}(u,v)$ selects the weight of the edge $(u,v)$, 
and its highest-energy node (i.e., $u$ if $f(u,v)<0$,
and $v$ otherwise, except when $v=z$), together with the weight to reach 
that highest energy node node from $u$.
Recall that algorithm $\Cyclealgo$ from \cref{sec:min_cycle}
traverses a tree-decomposition bottom-up, and for each encountered bag
$\Bag$ stores a map $\LD_{\Bag}$ such that $\LD_{\Bag}(u,v)$ is upper bounded by
the weight of the shortest $\Ushape$-shaped simple path $u\Path v$ (or simple cycle, if $u=v$).
Our algorithm $\Energynodesalgotw$ for determining all $0$-energy nodes is similar,
only that now $\LD_{\Bag}$ stores triplets $(a,b,c)$ where $a$ is the weight of a $\Ushape$-shaped path $P$,
$b$ is a highest-energy node of $P$, and $c$ the weight of a highest-energy prefix of $P$.
For two triplets $\alpha_1=(a_1,b_1, c_1),\alpha_2=(a_2,b_2, c_2)\in \Intstriplets$
corresponding to $\Ushape$-shaped paths $P_1$ and $P_2$,
$\Min(\alpha_1,\alpha_2)$ selects the path with the smallest weight,
and $\alpha_1\Plus \alpha_2$ determines the weight, a highest-energy node,
and the weight of a highest-energy prefix of the path $P_1\circ P_2$ (see \cref{fig:credit2}).

\input{fig_credit2}

\noindent{\bf Algorithm $\Energynodesalgotw$.}
The algorithm $\Energynodesalgotw$ for computing the set
of $0$-energy nodes in constant-treewidth graphs follows the same principle as $\Energynodesalgo$
for general graphs. 
It stores a map of edge weights $\Weightmap:E_2\rightarrow \Ints\cup\{\infty\}$, and initially
$\Weightmap(u,v)=\Weight_2(u,v)$ for each $(u,v)\in E_2$.
The algorithm performs a bottom-up pass, 
and computes in each bag the local distance map 
$\LD_{\Bag}:\Bag\times\Bag\rightarrow \Intstriplets$
that captures $\Ushape$-shaped $u\Path v$ paths, together with their highest-energy nodes.
When a non-positive cycle $C$ is found in some bag $\Bag$, the method $\Killcyclealgo$ is called
to modify the edges of a highest-energy node $w$ of $C$
and its incoming neighbors by updating the map $\Weightmap$.
These updates generally affect the distances between
the rest of the nodes in the graph, hence some local distance maps $\LD_{\Bag}$
need to be corrected.
However, each such edge modification only affects the local distance map
of bags that appear in a path from a bag $\Bag'$ to some
ancestor $\Bag''$ of $\Bag'$.
Instead of restarting the computation as in $\Energynodesalgo$,
the method $\Updatealgo$ is called to correct those local distance maps
along the path $\Bag'\Path\Bag''$.

\begin{algorithm}
\small
\DontPrintSemicolon
\caption{$\Energynodesalgotw$}\label{algo:energynodesalgotw}
\KwIn{A weighted graph $G_2=(V_2,E_2,\Weight_2)$ and a binary tree-decomposition $\Tree(G_2)$}
\KwOut{The set $\{v\in V_2\setminus\{z\}: \Energy(v)=0\}$}
\BlankLine
\tcp{Initialization}
Assign $X\leftarrow\emptyset$\\
\ForEach{$u,v\in V_2$}{
\eIf{$(u,v)\in E_2$}{
Assign $\Weightmap(u,v)\leftarrow \Weight_2(u,v)$
}
{
Assign $\Weightmap(u,v)\leftarrow \infty$
}
}
\tcp{Computation}
Apply a post-order traversal on $\Tree(G)$, and examine each bag $\Bag$ with children $\Bag_1, \Bag_2$\\
\Begin{
\ForEach{$u, v\in \Bag$}{
Assign $\LD_{\Bag}(u,v)\leftarrow\Min(\LD_{\Bag_1}(u,v), \LD_{\Bag_2}(u,v), g_{\Weightmap}(u,v))$\\
}
\uIf{$\Bag$ is the root bag of a node $x$}{
\ForEach{$u, v\in \Bag$}{
Assign $\LD'_{\Bag}(u,v)\leftarrow \Min(\LD_{\Bag}(u,v), \LD_{\Bag}(u,x)\Plus\LD_{\Bag}(x,v))$\\
}
Assign $\LD_{\Bag}\leftarrow \LD'_{\Bag}$\\
\uIf{$\exists u\in \Bag$ with $\LD_{\Bag}(u,u)=(a,b,c)$ where $a\leq 0$}{\label{line:if_negative}
Assign $X\leftarrow X\cup \{b\}$\label{line:add_zero_node}\\
Execute $\Killcyclealgo$ on $b$ and $\Bag$ 
}
}
}
\Return $X$
\end{algorithm}

\begin{algorithm}
\small
\SetAlgorithmName{Method}{method}
\DontPrintSemicolon
\caption{$\Killcyclealgo$}\label{algo:killcyclealgo}
\KwIn{A $0$-energy node $w$ and a bag $\Bag$ of $\Tree(G_2)$}
\KwOut{Updates the local distance function $\LD_{\Bag}$}
\BlankLine
\ForEach{edge $(x,w)\in E_2$}{
Assign $\Weightmap(x,z)\leftarrow \min(\Weight_2(x,w), \Weightmap(x,z))$\label{line:tw_mod1}\\
Assign $\Weightmap(x,w)\leftarrow \infty$\label{line:tw_mod2}\\
Assign $y\leftarrow \arg\max_{u\in\{x,w\}}\Level(u)$\label{line:y}\\
Let $\Bag'$ be the smallest-level ancestor of $\Bag_y$ examined by $\Energynodesalgotw$ so far\\
Execute $\Updatealgo$ on $\Bag_y$ and its ancestor $\Bag'$
}
\Return $\LD_{\Bag}$
\end{algorithm}

\begin{algorithm}
\small
\SetAlgorithmName{Method}{method}
\DontPrintSemicolon
\caption{$\Updatealgo$}\label{algo:updatealgo}
\KwIn{A bag $\Bag'$ and an ancestor $\Bag''$}
\KwOut{The local distances $\LD_{\Bag}$ along the path $\Bag'\Path \Bag''$}
\BlankLine
Traverse the path $\Bag'\Path \Bag''$ bottom-up, and examine each bag $\Bag$ with children $\Bag_1, \Bag_2$\\
\Begin{
\ForEach{$u, v\in \Bag$}{
Assign $\LD_{\Bag}(u,v)\leftarrow\Min(\LD_{\Bag_1}(u,v), \LD_{\Bag_2}(u,v), g_{\Weightmap}(u,v))$\\
}
\uIf{$\Bag$ is the root bag of a node $x$}{
\ForEach{$u, v\in \Bag$}{
Assign $\LD'_{\Bag}(u,v)\leftarrow \Min(\LD_{\Bag}(u,v), \LD_{\Bag}(u,x)\Plus\LD_{\Bag}(x,v))$\\
}
Assign $\LD_{\Bag} \leftarrow \LD'_{\Bag}$\\
\uIf{$\exists u\in \Bag$ with $\LD_{\Bag}(u,u)=(a,b,c)$ where $a\leq 0$}{\label{line:update_if_negative}
Assign $X\leftarrow X\cup \{b\}$\label{line:update_add_zero_node}\\
Execute $\Killcyclealgo$ on $b$ and $\Bag$ 
}
}
}
\end{algorithm}

The following lemma establishes the correctness of $\Energynodesalgotw$.
Similarly as for \cref{lem:energynodesalgo_cor1} and \cref{lem:energynodesalgo_cor2} 
we denote with $\Graphstruct^k$ the graph obtained by considering the edges $(u,v)$
for which $\Weightmap(u,v)<\infty$ when $|X|=k$.

\begin{lemma}\label{lem:energynodesalgotw_cor}
For every $v\in V\setminus\{z\}$ we have $v\in X$ iff $\Energy(v)=0$.
\end{lemma}
\begin{proof}
We only need to argue that $\Energynodesalgotw$ correctly computes the 
non-positive cycles in every $\Graphstruct^k$, as then the correctness follows
from the correctness Lemma~\ref{lem:energynodesalgo_cor1} and Lemma~ \ref{lem:energynodesalgo_cor2}
of $\Energynodesalgo$.
Since by \cref{rem:ushape_cycle} every cycle is a $\Ushape$-shaped path in some bag,
it suffices to argue that whenever $\Energynodesalgotw$ examines a bag $\Bag$
(either directly, or through $\Updatealgo$),
every $\Ushape$-shaped simple cycle in $\Bag$ has been considered by 
the algorithm. 
This is true if no calls to $\Killcyclealgo$ are made (\emph{if} block in line~\ref{line:if_negative}),
as then $\Energynodesalgotw$ is the same as $\Cyclealgo$, and hence it follows from \cref{lem:ushape}.

Now consider that $\Killcyclealgo$ is called and $\Bag'$ is the smallest-level bag examined by
$\Energynodesalgotw$ so far.
Let $w$ be the $0$-energy node, $x$ an incoming neighbor of $w$, and $y=\arg\max_{u\in\{x,w\}}\Level(u)$
(as in line~\ref{line:y} of $\Killcyclealgo$).
By the definition of $\Ushape$-shaped paths, the edge $(x,w)$ appears only in paths
that are $\Ushape$-shaped in bags along the path $\Bag_y\Path \Bag'$.
Hence, after setting $\Weightmap(x,w)=\infty$ (line~\ref{line:tw_mod2} of $\Killcyclealgo$),
it suffices to update the local distance maps of these bags.
Similarly, after setting $\Weightmap(x,z)\leftarrow \min(\Weight_2(x,w), \Weightmap(x,z))$ (line~\ref{line:tw_mod1} of $\Killcyclealgo$),
since $\Bag_z$ is the root of $\Tree(G_2)$,
it suffices to update the local distance maps in the bags along the path $\Bag_x\Path\Bag'$.
Either $x=y$, or, by the properties of tree-decompositions, $\Bag_x$ is an ancestor of $\Bag_y$.
Hence in either case $\Bag_x\Path \Bag'$ is a subpath of $\Bag_y\Path \Bag'$,
and both edge modifications in lines~\ref{line:tw_mod1} and \ref{line:tw_mod2} are handled
correctly by calling $\Updatealgo$ on $\Bag_y$ and its ancestor $\Bag'$.
The result follows.
\end{proof}

\begin{lemma}\label{lem:energynodesalgotw_compl}
Algorithm $\Energynodesalgotw$ runs in $O(n\cdot \log n)$ time and $O(n)$ space.
\end{lemma}
\begin{proof}
Let $h=O(\log n)$ be the height of $\Tree(G_2)$.
\begin{compactenum}
\item The method $\Updatealgo$ performs a constant number of operations
to each bag in the path $\Bag'\Path \Bag''$ where $\Bag''$ is ancestor of $\Bag'$,
hence each call to $\Updatealgo$ requires $O(h)$ time.
\item The method $\Killcyclealgo$ performs a constant number of operations locally
and one call to $\Updatealgo$ for each incoming edge of $w$. Hence if $w$ has 
$k_w$ incoming edges, $\Killcyclealgo$ requires $O(h\cdot k_w)$ time.
Since $\Killcyclealgo$ sets $\Weightmap(x,w)=\infty$ for all incoming edges of $w$,
the node $w$ will not appear in non-positive cycles thereafter.
\item The algorithm $\Energynodesalgotw$ is similar to $\Cyclealgo$
which runs in $O(n)$ time and space (\cref{lem:cyclealgo_complexity}).
The difference is in the additional \emph{if} block in line~\ref{line:if_negative}.
Since $\Killcyclealgo$ is called when a non-positive cycle is detected, 
it will be called at most once for each node $u\in V_2\setminus\{z\}$
(from either $\Energynodesalgotw$ or $\Updatealgo$).
It follows that the total time of $\Energynodesalgotw$ is
\[
O\left(n+\sum_u (h\cdot k_u)\right)=O(n+h\cdot |E_2|)=O(n\cdot \log n)
\]
where $k_u$ is the number of incoming edges of node $u$.
Since $\Killcyclealgo$ stores constant size of information in
each bag of $\Tree(G_2)$, the $O(n)$ space bound follows.
\end{compactenum}
\end{proof}

After the set $X$ of $0$-energy nodes has been computed, it remains
to execute one instance of the single-source shortest path problem on the graph $\Graphstruct^{|X|}$
(similarly as for our solution on general graphs).
It is known that single-source distances in tree-decompositions of constant
treewidth can be computed in $O(n)$ time~\cite{Chaudhuri95,Chatterjee15}.
We thus obtain the following theorem.

\begin{theorem}\label{them:energy_tw}
Let $G=(V,E,\Weight)$ be a weighted graph of $n$ nodes with constant treewidth.
The minimum initial credit value problem for $G$ can be solved in $O(n\cdot \log n)$ time
and $O(n)$ space.
\end{theorem}

%% file: fig_credit.tex
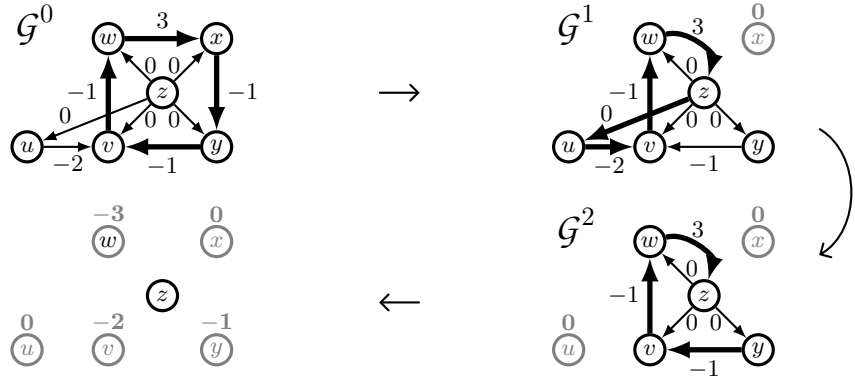
\begin{figure}[!h]
\centering
\begin{tikzpicture}[scale=0.9,
thick, >=latex,
pre/.style={<-,shorten >= 1pt, shorten <=1pt, thick},
post/.style={->,shorten >= 1pt, shorten <=1pt,  thick},
und/.style={very thick, draw=gray},
node/.style={circle, minimum size=4mm, draw=black!100, very thick, inner sep=1},
rootbag/.style={ellipse, minimum height=7mm,minimum width=14mm,draw=black!80, line width=2.5pt},
virt/.style={circle,draw=black!50,fill=black!20, opacity=0}]

\def \step{1.6}
\def \bend{0}

\def \xdisp{0}
\def \ydisp{0}
\def \uxoffset{0.4}

\node []		(G0)	at	(\xdisp+0.5, \ydisp+1.8)	{\Large $\Graphstruct^0$};	

\node	[node]	(u)	at	(0*\step\xdisp+\uxoffset, 0*\step+\ydisp)	{$u$};
\node	[node]	(v)	at	(1*\step+\xdisp, 0*\step+\ydisp)	{$v$};
\node	[node]	(w)	at	(1*\step+\xdisp, 1*\step+\ydisp)	{$w$};
\node	[node]	(x)	at	(2*\step+\xdisp, 1*\step+\ydisp)	{$x$};
\node	[node]	(y)	at	(2*\step+\xdisp, 0*\step+\ydisp)	{$y$};
\node	[node]	(z)	at	(3*\step/2+\xdisp, 1*\step/2+\ydisp)	{$z$};

\draw [draw=black,  thick, ->,] (u) to node[below]{$-2$} (v);
\draw [draw=black, line width=2, ->,] (v) to node[left]{$-1$} (w);
\draw [draw=black, line width=2, ->,] (w) to node[above]{$3$} (x);
\draw [draw=black, line width=2, ->,] (x) to node[right]{$-1$} (y);
\draw [draw=black, line width=2, ->,] (y) to node[below]{$-1$} (v);

\draw [draw=black,  thick, ->, bend right=\bend] (z) to node[above, pos=0.8]{$0$} (u);
\draw [draw=black,  thick, ->,] (z) to node[right]{$0$} (v);
\draw [draw=black,  thick, ->,] (z) to node[right]{$0$} (w);
\draw [draw=black,  thick, ->,] (z) to node[left]{$0$} (x);
\draw [draw=black,  thick, ->,] (z) to node[left]{$0$} (y);

\draw [draw=black,  thick, -angle 90,] (4+\step,\step/2) to (4+\step+0.6,\step/2);

\def \xdisp{8}
\def \ydisp{0}

\node []		(G1)	at	(\xdisp+0.5, \ydisp+1.8)	{\Large $\Graphstruct^1$};	

\node	[node]	(u)	at	(0*\step + \xdisp+\uxoffset, 0*\step+\ydisp)	{$u$};
\node	[node]	(v)	at	(1*\step+\xdisp, 0*\step+\ydisp)	{$v$};
\node	[node]	(w)	at	(1*\step+\xdisp, 1*\step+\ydisp)	{$w$};
\node	[]	(xe)	at	(2*\step+\xdisp, 1*\step+\ydisp+0.4)	{\textcolor{gray}{$\mathbf{0}$}};
\node	[node, draw=gray]	(x)	at	(2*\step+\xdisp, 1*\step+\ydisp)	{\textcolor{gray}{$x$}};
\node	[node]	(y)	at	(2*\step+\xdisp, 0*\step+\ydisp)	{$y$};
\node	[node]	(z)	at	(3*\step/2+\xdisp, 1*\step/2+\ydisp)	{$z$};

\draw [draw=black,  line width=2, ->,bend right=-60] (w) to node[above]{$3$} (z);
\draw [draw=black,  line width=2, ->,] (u) to node[below]{$-2$} (v);
\draw [draw=black,  line width=2, ->,] (v) to node[left]{$-1$} (w);
\draw [draw=black,  thick, ->,] (y) to node[below]{$-1$} (v);

\draw [draw=black,  line width=2, ->, bend right=\bend] (z) to node[above, pos=0.8]{$0$} (u);
\draw [draw=black,  thick, ->,] (z) to node[right]{$0$} (v);
\draw [draw=black,  thick, ->,] (z) to node[right]{$0$} (w);
\draw [draw=black,  thick, ->,] (z) to node[left]{$0$} (y);


\def \xdisp{8}
\def \ydisp{-3}

\node []		(G0)	at	(\xdisp+0.5, \ydisp+1.8)	{\Large $\Graphstruct^2$};	

\node	[]	(ue)	at	(0*\step + \xdisp+\uxoffset, 0*\step+\ydisp+0.4)	{\textcolor{gray}{$\mathbf{0}$}};
\node	[node, draw=gray]	(u)	at	(0*\step + \xdisp+\uxoffset, 0*\step+\ydisp)	{\textcolor{gray}{$u$}};
\node	[node]	(v)	at	(1*\step+\xdisp, 0*\step+\ydisp)	{$v$};
\node	[node]	(w)	at	(1*\step+\xdisp, 1*\step+\ydisp)	{$w$};
\node	[]	(xe)	at	(2*\step+\xdisp, 1*\step+\ydisp+0.4)	{\textcolor{gray}{$\mathbf{0}$}};
\node	[node, draw=gray]	(x)	at	(2*\step+\xdisp, 1*\step+\ydisp)	{\textcolor{gray}{$x$}};
\node	[node]	(y)	at	(2*\step+\xdisp, 0*\step+\ydisp)	{$y$};
\node	[node]	(z)	at	(3*\step/2+\xdisp, 1*\step/2+\ydisp)	{$z$};

\draw [draw=black, line width=2, ->,bend right=-60] (w) to node[above]{$3$} (z);
\draw [draw=black, line width=2, ->,] (v) to node[left]{$-1$} (w);
\draw [draw=black, line width=2, ->,] (y) to node[below]{$-1$} (v);

\draw [draw=black,  thick, ->,] (z) to node[right]{$0$} (v);
\draw [draw=black,  thick, ->,] (z) to node[right]{$0$} (w);
\draw [draw=black,  thick, ->,] (z) to node[left]{$0$} (y);

\draw [draw=black,  thick, angle 90-,] (4+\step,-2.3) to (4+\step+0.6,-2.3);

\def \xdisp{0}
\def \ydisp{-3}

\node	[]	(ue)	at	(0*\step + \xdisp+\uxoffset, 0*\step+\ydisp+0.4)	{\textcolor{gray}{$\mathbf{0}$}};
\node	[node, draw=gray]	(u)	at	(0*\step + \xdisp+\uxoffset, 0*\step+\ydisp)	{\textcolor{gray}{$u$}};
\node	[]	(ve)	at	(1*\step+\xdisp, 0*\step+\ydisp+0.4)	{\textcolor{gray}{$\mathbf{-2}$}};
\node	[node, draw=gray]	(v)	at	(1*\step+\xdisp, 0*\step+\ydisp)	{\textcolor{gray}{$v$}};
\node	[]	(we)	at	(1*\step+\xdisp, 1*\step+\ydisp+0.4)	{\textcolor{gray}{$\mathbf{-3}$}};
\node	[node, draw=gray]	(w)	at	(1*\step+\xdisp, 1*\step+\ydisp)	{$w$};
\node	[]	(xe)	at	(2*\step+\xdisp, 1*\step+\ydisp+0.4)	{\textcolor{gray}{$\mathbf{0}$}};
\node	[node, draw=gray]	(x)	at	(2*\step+\xdisp, 1*\step+\ydisp)	{\textcolor{gray}{$x$}};
\node	[]	(ye)	at	(2*\step+\xdisp, 0*\step+\ydisp+0.4)	{\textcolor{gray}{$\mathbf{-1}$}};
\node	[node, draw=gray]	(y)	at	(2*\step+\xdisp, 0*\step+\ydisp)	{\textcolor{gray}{$y$}};
\node	[node]	(z)	at	(3*\step/2+\xdisp, 1*\step/2+\ydisp)	{$z$};

\draw [draw=black,  thick, -angle 90, bend left=60] (10.5+\step,\step/6) to (10.5+\step,-1.6);

\end{tikzpicture}
\caption{
Solving the value problem using operations on the graph $\Graphstruct$. 
Initially we examine $\Graphstruct^0$,
and a non-positive cycle is found (boldface edges) with highest-energy node $x$.
Thus $\Energy(x)=0$, and we proceed with $\Graphstruct^1$, to discover $\Energy(u)=0$.
In $\Graphstruct^2$ all cycles are positive, and the energy of each remaining node
is minus its distance to $z$.
}\label{fig:credit}
\end{figure}

%% file: fig_credit2.tex
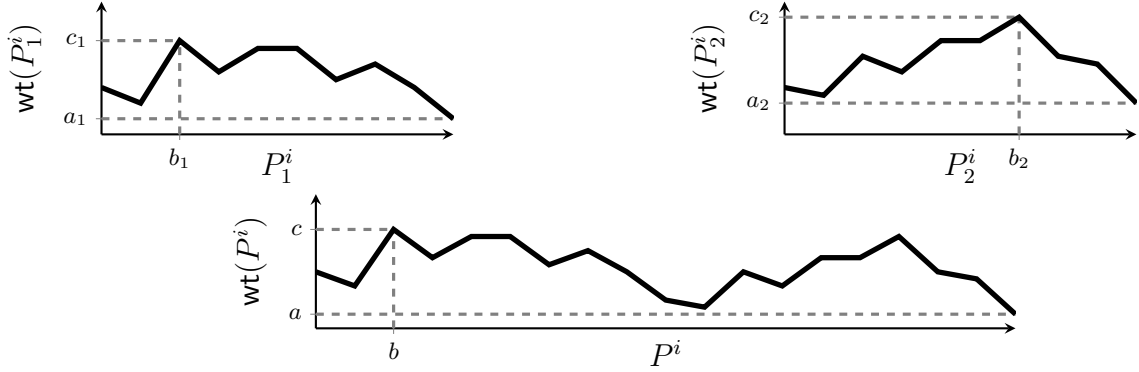
\begin{figure}
\def \scale{0.8}
\centering
\begin{subfigure}{0.45\textwidth}
\begin{tikzpicture}
	\begin{axis}[scale=\scale,
	axis line style = thick,
	compat=newest,
	 xlabel style={yshift=1.7em},
	unit vector ratio*=5 1 1,
	width=\textwidth,
		axis lines=left,
		ymax=11,
		ymin=-6,
		xtick=\empty, ytick=\empty,
		extra x ticks={2},
		extra x tick labels={$b_1$},
		extra y ticks={6,-4},
		extra y tick labels={$c_1$, $a_1$},
		xlabel={\large $P_1^i$},
		ylabel={\large $\Weight(P_1^i)$}]
		\draw[dashed,gray, very thick] (axis cs: 2,-6)--(axis cs: 2,6);
	\draw[dashed,gray, very thick] (axis cs: 0,6)--(axis cs: 2,6);
	\draw[dashed,gray, very thick] (axis cs: 0,-4)--(axis cs: 9,-4);
	\addplot[color=black,mark=.,line width=2] coordinates {
		(0,0)
		(1,-2)
		(2,6)
		(3,2)
		(4,5)
		(5,5)
		(6,1)
		(7,3)
		(8,-0)
		(9,-4)
	};
	\end{axis}
	\end{tikzpicture}
	\end{subfigure}%
	\hfill
	\begin{subfigure}{0.45\textwidth}
	\begin{tikzpicture}
	\begin{axis}[
	scale=\scale,
	axis line style = thick,
	compat=newest,
	unit vector ratio*=5 1 1,
	width=\textwidth,
		axis lines=left,
		xlabel style={yshift=1.7em},
		ymax=11,
		ymin=-6,
		xtick=\empty, ytick=\empty,
		extra x ticks={6},
		extra x tick labels={$b_2$},
		extra y ticks={9,-2},
		extra y tick labels={$c_2$, $a_2$},
		xlabel={\large $P_2^i$},
		ylabel={\large $\Weight(P_2^i)$}]
		\draw[dashed,gray, very thick] (axis cs: 6,-6)--(axis cs: 6,9);
	\draw[dashed,gray, very thick] (axis cs: 0,9)--(axis cs: 6,9);
	\draw[dashed,gray, very thick] (axis cs: 0,-2)--(axis cs: 9,-2);
	\addplot[color=black,mark=.,line width=2] coordinates {
		(0,0)
		(1,-1)
		(2,4)
		(3,2)
		(4,6)
		(5,6)
		(6,9)
		(7,4)
		(8,3)
		(9,-2)
	};
	\end{axis}
	\end{tikzpicture}
	\end{subfigure}%
	\hfill
	\begin{subfigure}{\textwidth}
	\centering
	\begin{tikzpicture}
	\begin{axis}[
	scale=\scale,
	axis line style = thick,
	compat=newest,
	unit vector ratio*=5.5 1 1,
	width=0.8\textwidth,
		axis lines=left,
		xlabel style={yshift=1.7em},
		ymax=11,
		ymin=-8,
		xtick=\empty, ytick=\empty,
		extra x ticks={2},
		extra x tick labels={$b$},
		extra y ticks={6, -6},
		extra y tick labels={$c$, $a$},
		xlabel={\large $P^i$},
		ylabel={\large $\Weight(P^i)$}]
		\draw[dashed,gray, very thick] (axis cs: 2,-8)--(axis cs: 2,6);
	\draw[dashed,gray, very thick] (axis cs: 0,6)--(axis cs: 2,6);
	\draw[dashed,gray, very thick] (axis cs: 0,-6)--(axis cs: 18,-6);
	\addplot[color=black,mark=.,line width=2] coordinates {
	(0,0)
		(1,-2)
		(2,6)
		(3,2)
		(4,5)
		(5,5)
		(6,1)
		(7,3)
		(8,-0)
		(9,-4)
		(10,-5)
		(11,0)
		(12,-2)
		(13,2)
		(14,2)
		(15,5)
		(16,0)
		(17,-1)
		(18,-6)
	};
	\end{axis}
	\end{tikzpicture}
	\end{subfigure}
\caption{
Illustration of the $\alpha_1\Plus \alpha_2$ operation, corresponding to concatenating paths $P_1$ and $P_2$.
The path $P_j^i$ denotes the $i$-th prefix of $P_j$. We have $P=P_1\circ P_2$, and the corresponding tripplet
$\alpha=(a,b,c)$ denotes the weight $a$ of $P$, its highest-energy node $b$, and the weight $c$ of a highest-energy prefix.
}\label{fig:credit2}
\end{figure}

%% file: results.tex
\section{Experimental Results}\label{sec:results}
In the current section we report on preliminary experimental evaluation of our algorithms,
and compare them to existing methods.
Our algorithm for the minimum mean cycle problem provides improvement
for constant-treewidth graphs, and has thus been evaluated on
low-treewidth graphs obtained from the control-flow graphs of programs.
For the minimum initial credit problem, we have implemented our algorithm
for arbitrary graphs, thus the benchmarks used in this case are general graphs
(i.e., not constant-treewidth graphs).

\subsection{Minimum mean cycle}
We have implemented our approximation algorithm for the minimum mean cycle problem,
and we let the algorithm run for as many iterations until a minimum mean cycle was discovered, 
instead of terminating after $O(\log(n/\epsilon))$ iterations required by \cref{them:min_mean_approx}.
We have tested its performance in running time and space against
six other minimum mean cycle algorithms from \cref{tab:algos}
in control-flow graphs of programs.
The algorithms of Burns and Lawler solve the more general ratio cycle problem,
and have been adapted to the mean cycle problem as in~\cite{DIG98}.

\begin{table}[h]
\small
\renewcommand{\arraystretch}{1.3}
\centerline{
\begin{tabular}{|c||c|c|c|c|c|c|}
\hline
& Madani~\cite{Mad02} &  Burns~\cite{B91} & Lawler~\cite{L76} & Dasdan-Gupta~\cite{Dasdan98} & Hartmann-Orlin~\cite{Hartmann93} & Karp~\cite{Karp78}\\
\hline
\hline
Time & $O(n^2)$ & $O(n^3)$ & $O(n^2\cdot \log (n\cdot W))$  & $O(n^2)$ & $O(n^2 )$ &   $O(n^2)$\\
\hline
Space & $O(n)$ & $O(n)$ & $O(n)$ &  $O(n^2)$ & $O(n^2)$ & $O(n^2)$\\
\hline
\end{tabular}
}
\caption{Asymptotic complexity of compared minimum mean cycle algorithms.}\label{tab:algos}
\end{table}

\noindent{\bf Setup.}
The algorithms were executed on control-flow graphs of methods of programs
from the DaCapo benchmark suit~\cite{Blackburn06},
obtained using the Soot framework~\cite{Soot}.
For each benchmark we focused on graphs of at least $500$ nodes.
This supplied a set of medium sized graphs (between $500$ and $1300$ nodes), 
in which integer weights were assigned uniformly at random in the range $\{-10^3,\dots, 10^3\}$.
Memory usage was measured with \cite{jamm}.

\noindent{\bf Results.}
\cref{fig:fig_cycle} shows the average time and space performance of the examined algorithms
(bars that exceeded the maximum value in the y-axis have been truncated).
Our algorithm has much smaller running time than each other algorithm, in almost all cases.
In terms of space, our algorithm significantly outperforms all others, except for the algorithms of Lawler, Burns, and Madani.
Both ours and these three algorithms have linear space complexity, but ours also suffers some
constant factor overhead from the tree-decomposition 
(i.e., the same node generally appears in multiple bags).
Note that the strong performance of these three algorithms in space is followed by
poor performance in running time.

\begin{figure}
\centerline{
\includegraphics[scale=0.37]{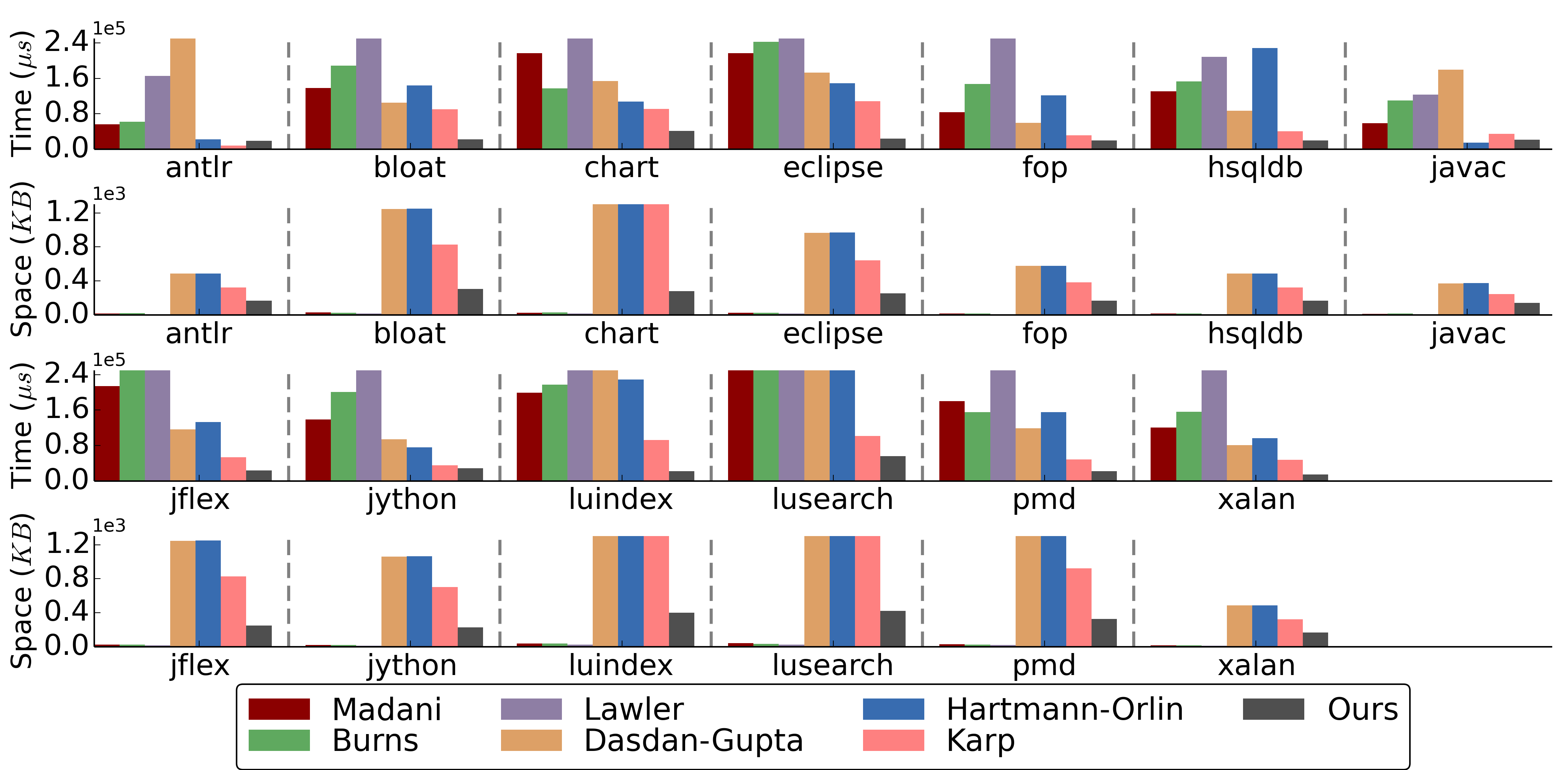}
}
\caption{Average performance of minimum mean cycle algorithms.}\label{fig:fig_cycle}
\end{figure}

\begin{table}[h]
\small
\renewcommand{\arraystretch}{1.2}
\centerline{
\begin{tabular}{|c||c|c|c|c|c|c|c|c|}
\hline
& ~Madani~&~~Burns~~&~Lawler~&~Dasdan-Gupta~&~Hartmann-Orlin~&~~Karp~~&~~\textbf{Ours}~~\\
\hline
\hline
antlr & 55814 & 61571 & 165789 & 284996 & 21893 & 7824 & \textbf{18402} \\
\hline
bloat & 138416 & 188356 & 350302 & 105145 & 144171 & 89949 & \textbf{22391} \\
\hline
chart & 216962 & 137112 & 573767 & 154062 & 107229 & 90717 & \textbf{40890} \\
\hline
eclipse & 216859 & 242323 & 667869 & 172792 & 148523 & 107864 & \textbf{23486} \\
\hline
fop & 83080 & 147384 & 406371 & 59176 & 121742 & 31557 & \textbf{19306} \\
\hline
hsqldb & 131041 & 153232 & 208328 & 86840 & 228632 & 40486 & \textbf{19957} \\
\hline
javac & 58443 & 110149 & 122996 & 179647 & 14719 & 34188 & \textbf{20874} \\
\hline
jflex & 214297 & 524822 & 554093 & 116820 & 133323 & 53329 & \textbf{23860} \\
\hline
jython & 139106 & 200922 & 503766 & 94052 & 75569 & 34864 & \textbf{28760} \\
\hline
luindex & 199650 & 217980 & 1240411 & 274319 & 228856 & 92379 & \textbf{22142} \\
\hline
lusearch & 433211 & 447280 & 1180051 & 263467 & 333297 & 101584 & \textbf{55652} \\
\hline
pmd & 180551 & 155118 & 585315 & 118578 & 155682 & 48326 & \textbf{21978} \\
\hline
xalan & 120897 & 156111 & 394458 & 81103 & 96873 & 47996 & \textbf{14493} \\
\hline
\end{tabular}
}
\caption{The time performance of \cref{fig:fig_cycle} (in $\mu s$).}\label{tab:results_time}
\end{table}

\begin{table}[h]
\small
\renewcommand{\arraystretch}{1.2}
\centerline{
\begin{tabular}{|c||c|c|c|c|c|c|c|c|}
\hline
&~Madani~&~~Burns~~&~Lawler~&~Dasdan-Gupta~&~Hartmann-Orlin~&~~Karp~~&~~\textbf{Ours}~~ \\
\hline
\hline
antlr & 16805 & 21018 & 11144 & 486435 & 489176 & 322384 & \textbf{168648} \\
\hline
bloat & 29723 & 24500 & 19458 & 1245272 & 1249444 & 826645 & \textbf{306026} \\
\hline
chart & 27130 & 30567 & 18172 & 2025448 & 2029294 & 1347048 & \textbf{278586} \\
\hline
eclipse & 24215 & 26488 & 16293 & 965063 & 968595 & 640720 & \textbf{254393} \\
\hline
fop & 16845 & 17975 & 11052 & 576174 & 578646 & 382338 & \textbf{169738} \\
\hline
hsqldb & 16798 & 19309 & 11144 & 486435 & 489096 & 322384 & \textbf{168648} \\
\hline
javac & 14681 & 17047 & 9664 & 372697 & 375453 & 247019 & \textbf{144721} \\
\hline
jflex & 24561 & 26946 & 16322 & 1244495 & 1248036 & 826743 & \textbf{251549} \\
\hline
jython & 22518 & 23337 & 14899 & 1059291 & 1062570 & 703581 & \textbf{228207} \\
\hline
luindex & 39309 & 40223 & 25604 & 3521607 & 3526792 & 2342833 & \textbf{399076} \\
\hline
lusearch & 41488 & 33350 & 26991 & 3387914 & 3393343 & 2253403 & \textbf{422679} \\
\hline
pmd & 32204 & 24481 & 21021 & 1391551 & 1395786 & 923975 & \textbf{326137} \\
\hline
xalan & 16798 & 17763 & 11144 & 486435 & 489102 & 322384 & \textbf{168648} \\
\hline
\end{tabular}
}
\caption{The space performance of \cref{fig:fig_cycle} (in KB).}\label{tab:results_space}
\end{table}

\subsection{Minimum initial credit}
We have implemented our algorithm for the minimum initial credit problem on general graphs
and experimentally evaluated its performance on a subset of benchmark weighted graphs
from the DIMACS implementation challenges~\cite{dimacs}.
Our algorithm was tested against the existing method of~\cite{Bouyer08}.
The direct implementation of the algorithm of~\cite{Bouyer08} performed poorly,
and for this we also implemented an optimized version
(using techniques such as caching of intermediate results and early loop termination).
Note that we compare algorithms for general graphs, without the low-treewidth restriction.

\noindent{\bf Setup.}
For each input graph we first computed its minimum mean value $\mu^{\ast}$ using Karp's algorithm,
and then subtracted $\mu^{\ast}$ from the weight of each edge to ensure that
at least one non-positive cycle exists (thus the energies are finite).

\noindent{\bf Results.}
\cref{fig:fig_credit} depicts the running time of the algorithm of~\cite{Bouyer08} (with and without optimizations) vs our algorithm. 
A timeout was forced at $10^{10} \mu s$.
Our algorithm is orders of magnitude faster, and scales better than the existing method.

\begin{figure}
\centerline{
\includegraphics[scale=0.6]{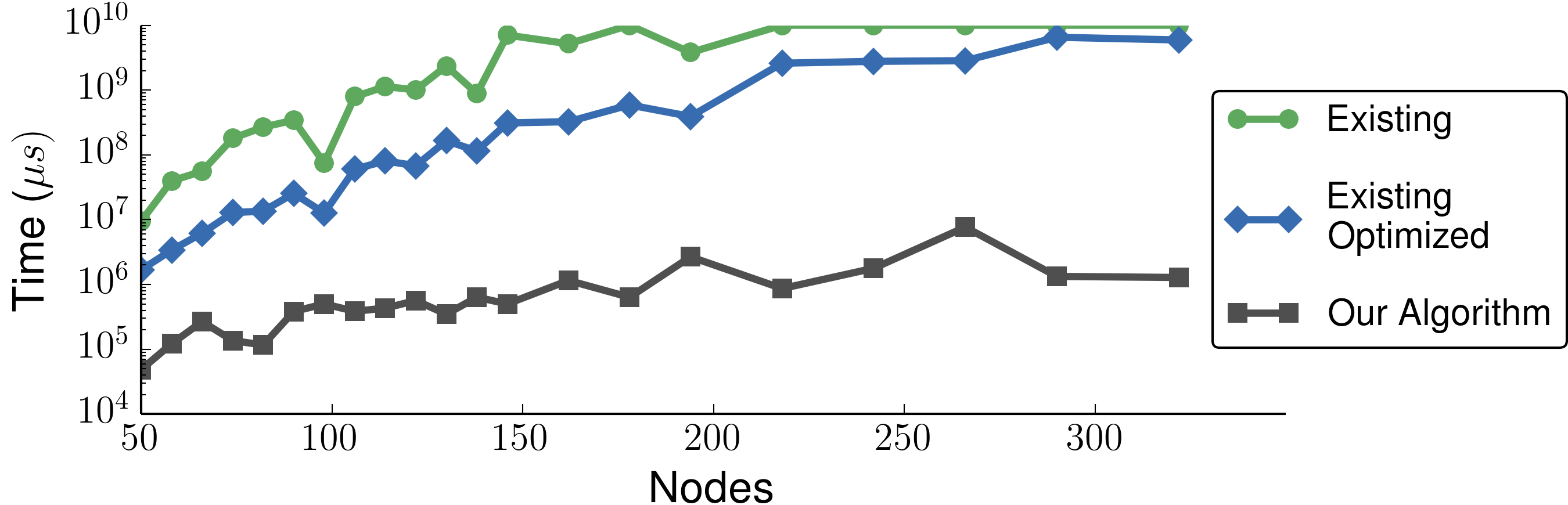}
}
\caption{Comparison of running times for the minimum initial credit problem.}\label{fig:fig_credit}
\end{figure}

\begin{table}[h]
\small
\renewcommand{\arraystretch}{1.3}
\centerline{
\begin{tabular}{|c||c|c|c|}
\hline
$n$ &~~~~~~Existing~~~~~~&~~Existing Optimized~~&~~~~~~~~\textbf{Ours}~~~~~~~~\\
\hline
\hline
50 & 9453565 & 1680924 & \textbf{48635} \\
\hline
58 & 39744129 & 3394193 & \textbf{121774} \\
\hline
66 & 55766874 & 6201044 & \textbf{267825} \\
\hline
74 & 180080064 & 12833610 & \textbf{136239} \\
\hline
82 & 267993314 & 13563936 & \textbf{116518} \\
\hline
90 & 342779026 & 25453589 & \textbf{383292} \\
\hline
98 & 74622910 & 12648395 & \textbf{501365} \\
\hline
106 & 791441986 & 60294150 & \textbf{385799} \\
\hline
114 & 1133055323 & 80584700 & \textbf{432290} \\
\hline
122 & 1004898322 & 67982455 & \textbf{564838} \\
\hline
130 & 2354354250 & 165193753 & \textbf{348112} \\
\hline
138 & 881117317 & 114743182 & \textbf{636481} \\
\hline
146 & 7050113907 & 311146051 & \textbf{501314} \\
\hline
162 & 5179877563 & 324877384 & \textbf{1154447} \\
\hline
178 & Timeout & 589873640 & \textbf{635155} \\
\hline
194 & 3799301931 & 391240954 & \textbf{2672127} \\
\hline
218 & Timeout & 2596083382 & \textbf{866213} \\
\hline
242 & Timeout & 2774469734 & \textbf{1779512} \\
\hline
266 & Timeout & 2839496222 & \textbf{7676638} \\
\hline
290 & Timeout & 6526762301 & \textbf{1332403} \\
\hline
322 & Timeout & 5929433611 & \textbf{1282258} \\
\hline
\end{tabular}
}
\caption{The time performance of \cref{fig:fig_credit} in $\mu s$.}\label{tab:results_time2}
\end{table}

%% file: main.bbl
\begin{thebibliography}{10}
\providecommand{\url}[1]{\texttt{#1}}
\providecommand{\urlprefix}{URL }

\bibitem{dimacs}
{DIMACS} implementation challenges, \url{http://dimacs.rutgers.edu/Challenges/}

\bibitem{Orna}
Almagor, S., Boker, U., Kupferman, O.: Formalizing and reasoning about quality.
  In: {ICALP}. LNCS, Springer (2013)

\bibitem{Blackburn06}
Blackburn, S.M., Garner, R., Hoffmann, C., Khang, A.M., McKinley, K.S.,
  Bentzur, R., Diwan, A., Feinberg, D., Frampton, D., Guyer, S.Z., Hirzel, M.,
  Hosking, A., Jump, M., Lee, H., Moss, J.E.B., Phansalkar, A., Stefanovi\'{c},
  D., VanDrunen, T., von Dincklage, D., Wiedermann, B.: The {DaCapo}
  benchmarks: Java benchmarking development and analysis. In: OOPSLA. ACM
  (2006)

\bibitem{BCHJ10}
Bloem, R., Chatterjee, K., Henzinger, T.A., Jobstmann, B.: Better quality in
  synthesis through quantitative objectives. In: {CAV}. LNCS, Springer (2015)

\bibitem{BGHJ11}
Bloem, R., Greimel, K., Henzinger, T.A., Jobstmann, B.: Synthesizing robust
  systems. In: {FMCAD} (2009)

\bibitem{Bodlaender93}
Bodlaender, H.L.: A tourist guide through treewidth. Acta Cybern.  (1993)

\bibitem{Bodlaender05}
Bodlaender, H.: Discovering treewidth. In: SOFSEM: Theory and Practice of
  Computer Science. LNCS, Springer (2005)

\bibitem{Bodlaender95}
Bodlaender, H., Hagerup, T.: Parallel algorithms with optimal speedup for
  bounded treewidth. In: ICALP. LNCS, Springer (1995)

\bibitem{BCHK11}
Boker, U., Chatterjee, K., Henzinger, T.A., Kupferman, O.: Temporal
  specifications with accumulative values. In: LICS (2011)

\bibitem{Bouyer08}
Bouyer, P., Fahrenberg, U., Larsen, K.G., Markey, N., Srba, J.: Infinite runs
  in weighted timed automata with energy constraints. In: {FORMATS}. LNCS,
  Springer (2008)

\bibitem{Patricia}
Bouyer, P., Markey, N., Matteplackel, R.M.: Averaging in {LTL}. In: {CONCUR}.
  LNCS, Springer (2014)

\bibitem{jamm}
Brosius, D.: Java agent for memory measurements,
  \url{https://github.com/jbellis/jamm}

\bibitem{B91}
Burns, S.M.: Performance analysis and optimization of asynchronous circuits.
  Tech. rep. (1991)

\bibitem{CHR13}
Cerny, P., Henzinger, T.A., Radhakrishna, A.: Quantitative abstraction
  refinement. In: POPL. ACM (2013)

\bibitem{CL13}
Chatterjee, K., Lacki, J.: Faster algorithms for {Markov} decision processes
  with low treewidth. In: CAV. LNCS, Springer (2013)

\bibitem{CDHRR10}
Chatterjee, K., Doyen, L., Edelsbrunner, H., Henzinger, T.A., Rannou, P.:
  Mean-payoff automaton expressions  (2010)

\bibitem{CDH11}
Chatterjee, K., Doyen, L., Henzinger, T.A.: Expressiveness and closure
  properties for quantitative languages. LMCS  (2010)

\bibitem{CDH10}
Chatterjee, K., Doyen, L., Henzinger, T.A.: Quantitative languages. Trans.
  Comput. Log.  (2010)

\bibitem{Chatterjee15}
Chatterjee, K., Goyal, P., Ibsen-Jensen, R., Pavlogiannis, A.: Faster
  algorithms for algebraic path properties in recursive state machines with
  constant treewidth. In: POPL (2015)

\bibitem{CHKLR14}
Chatterjee, K., Henzinger, M., Krinninger, S., Loitzenbauer, V., Raskin, M.A.:
  Approximating the minimum cycle mean. Theor. Comput. Sci.

\bibitem{CHJS15}
Chatterjee, K., Henzinger, T.A., Jobstmann, B., Singh, R.: Measuring and
  synthesizing systems in probabilistic environments. In: JACM (2015)

\bibitem{CHO14}
Chatterjee, K., Henzinger, T.A., Otop, J.: Nested weighted automata. Tech.
  rep., IST Austria (2014)

\bibitem{CV12}
Chatterjee, K., Velner, Y.: Mean-payoff pushdown games. In: LICS. IEEE Computer
  Society (2012)

\bibitem{Chaudhuri95}
Chaudhuri, S., Zaroliagis, C.D.: {Shortest Paths in Digraphs of Small
  Treewidth. Part I: Sequential Algorithms}. Algorithmica  (1995)

\bibitem{C90}
Courcelle, B.: The monadic second-order logic of graphs. i. recognizable sets
  of finite graphs. Inf. Comput.  (1990)

\bibitem{Dasdan98}
Dasdan, A., Gupta, R.: Faster maximum and minimum mean cycle algorithms for
  system-performance analysis. IEEE Transactions on Computer-Aided Design of
  Integrated Circuits and Systems  (1998)

\bibitem{DIG98}
Dasdan, A., Irani, S.S., Gupta, R.K.: An experimental study of minimum mean
  cycle algorithms. Tech. rep. (1998)

\bibitem{WeightedHandbook}
Droste, M., Kuich, W., Vogler, H.: Handbook of Weighted Automata. Springer
  (2009)

\bibitem{Droste}
Droste, M., Meinecke, I.: Weighted automata and weighted {MSO} logics for
  average and long-time behaviors. Inf. Comput.  220 (2012)

\bibitem{Elberfeld10}
Elberfeld, M., Jakoby, A., Tantau, T.: Logspace versions of the theorems of
  {B}odlaender and {C}ourcelle. In: FOCS. IEEE Computer Society (2010)

\bibitem{Gustedt02}
Gustedt, J., Mhle, O., Telle, J.: The treewidth of java programs. In: Algorithm
  Engineering and Experiments. LNCS, Springer (2002)

\bibitem{Halin76}
Halin, R.: S-functions for graphs. Journal of Geometry  (1976)

\bibitem{Hartmann93}
Hartmann, M., Orlin, J.B.: Finding minimum cost to time ratio cycles with small
  integral transit times. NETWORKS  23 (1993)

\bibitem{HO14}
Henzinger, T.A., Otop, J.: From model checking to model measuring. In:
  {CONCUR}. LNCS, Springer (2013)

\bibitem{Karp78}
Karp, R.M.: A characterization of the minimum cycle mean in a digraph. Discrete
  Mathematics  (1978)

\bibitem{Kwek03}
Kwek, S., Mehlhorn, K.: Optimal search for rationals. Inf. Process. Lett.
  86(1) (2003)

\bibitem{L76}
Lawler, E.: Combinatorial Optimization: Networks and Matroids. Saunders College
  Publishing (1976)

\bibitem{Mad02}
Madani, O.: Polynomial value iteration algorithms for deterministic {MDPs}. In:
  UAI. Morgan Kaufmann Publishers (2002)

\bibitem{Obdrzalek03}
Obdrz{\'a}lek, J.: Fast mu-calculus model checking when tree-width is bounded.
  In: CAV. LNCS, Springer (2003)

\bibitem{OA92}
Orlin, J.B., Ahuja, R.K.: New scaling algorithms for the assignment and minimum
  mean cycle problems. Math. Program.  (1992)

\bibitem{Papadimitriou79}
Papadimitriou, C.H.: Efficient search for rationals. IPL  (1979)

\bibitem{Robertson84}
Robertson, N., Seymour, P.: Graph minors. iii. planar tree-width. Journal of
  Combinatorial Theory, Series B  (1984)

\bibitem{Tarjan72}
Tarjan, R.: Depth-first search and linear graph algorithms. SIAM Journal on
  Computing  (1972)

\bibitem{Thorup98}
Thorup, M.: {All Structured Programs Have Small Tree Width and Good Register
  Allocation}. Inf. Comput.  (1998)

\bibitem{Soot}
Vall{\'e}e-Rai, R., Co, P., Gagnon, E., Hendren, L., Lam, P., Sundaresan, V.:
  Soot - a java bytecode optimization framework. In: CASCON '99. IBM Press
  (1999)

\bibitem{Vel12}
Velner, Y.: The complexity of mean-payoff automaton expression. In: {ICALP}.
  LNCS, Springer (2012)

\end{thebibliography}
